\documentclass{siamltex}

\usepackage{amsmath}
\usepackage{amssymb}
\usepackage{cite}
\usepackage{enumitem}
\usepackage{graphicx}
\usepackage{ifthen}
\usepackage{subfigure}
\usepackage{tabularx}
\usepackage{booktabs}
\usepackage{microtype}

\newcommand{\newparentheses}[3]{%
  \expandafter\newcommand\csname #1\endcsname[1]{#2##1#3}%
  \expandafter\newcommand\csname #1L\endcsname[1]{\bigl#2##1\bigr#3}%
  \expandafter\newcommand\csname #1XL\endcsname[1]{\Bigl#2##1\Bigr#3}%
  \expandafter\newcommand\csname #1V\endcsname[1]{\left#2##1\right#3}}

\newparentheses{parens}{(}{)}
\newparentheses{set}{\{}{\}}
\newparentheses{size}{|}{|}

\makeatletter
\newcommand{\onenewattribute}[3]{%
  \@ifundefined{#1}{\let\@@def\newcommand}{\let\@@def\renewcommand}%
  \expandafter\@@def\csname #1\endcsname[2][]{%
    \ifthenelse{\equal{##1}{}}%
    {#2\csname #3\endcsname{##2}}%
    {#2_{##1}\csname #3\endcsname{##2}}}}
\newcommand{\newattribute}[2]{%
  \onenewattribute{#1}{#2}{parens}%
  \onenewattribute{#1L}{#2}{parensL}%
  \onenewattribute{#1XL}{#2}{parensXL}%
  \onenewattribute{#1V}{#2}{parensV}}
\makeatother

\newattribute{Oh}{\mathrm{O}}
\newattribute{Thehta}{\Theta}
\newattribute{oh}{\mathrm{o}}

\newattribute{alg}{\textsc{Maf}}
\newattribute{aalg}{\textsc{Maaf}}
\newattribute{balg}{\textsc{Refine}}

\newattribute{interval}{I}

\newcommand{\subtree}[2][]{%
  \ifthenelse{\equal{#1}{}}%
  {T(#2)}%
  {#1(#2)}}
\newcommand{\induced}[2][]{%
  \ifthenelse{\equal{#1}{}}%
  {T|#2}%
  {#1|#2}}
\newattribute{maf}{m}
\newattribute{maaf}{\tilde{m}}
\newattribute{ecut}{e}
\newattribute{aecut}{\tilde{e}}
\newcommand{\reach}[1][]{%
  \ifthenelse{\equal{#1}{}}%
  {\sim}%
  {\sim_{#1}}}
\newcommand{\noreach}[1][]{%
  \ifthenelse{\equal{#1}{}}%
  {\nsim}%
  {\nsim_{#1}}}
\newcommand{\edge}[1]{e_{#1}}
\newcommand{\triple}[2]{#1|#2}
\newcommand{\parent}[1]{p_{#1}}
\newcommand{\cyc}[1]{G_{#1}}
\newcommand{\ecyc}[1]{G_{#1}^\ast}

\newattribute{dspr}{d_{\textit{SPR}}}
\newattribute{dhyb}{\textit{hyb}}

\newattribute{indeg}{\mathrm{deg}_{\mathrm{in}}}

\newtheorem{observation}{Observation}{\bfseries}{\itshape}

\begin{document}


\title{Fixed-Parameter Algorithms for Maximum Agreement Forests}

\author{%
  Chris Whidden\footnotemark[2]\ \footnotemark[5]%
  \and%
  Robert G. Beiko\footnotemark[3]\ \footnotemark[5]%
  \and%
  Norbert Zeh\footnotemark[4]\ \footnotemark[5]}

\renewcommand{\thefootnote}{\fnsymbol{footnote}}

\footnotetext[2]{Supported by doctoral scholarships from NSERC and the Killam
Trust.}
\footnotetext[3]{Canada Research Chair; supported in part by NSERC and Genome
Atlantic.}
\footnotetext[4]{Canada Research Chair; supported in part by NSERC.}
\footnotetext[5]{Address: Faculty of Computer Science, Dalhousie University,
  Halifax, Nova Scotia, Canada.  Email: \{whidden, beiko, nzeh\}@cs.dal.ca.}

\renewcommand{\thefootnote}{\arabic{footnote}}

\maketitle

\pagestyle{myheadings}
\thispagestyle{plain}
\markboth{C.~WHIDDEN, R.~G.~BEIKO, AND N.~ZEH}{FIXED-PARAMETER ALGORITHMS
FOR MAXIMUM AGREEMENT FORESTS}


\begin{abstract}
  We present new and improved fixed-parameter algorithms for computing
  maximum agreement forests (MAFs) of pairs of rooted binary phylogenetic trees.
  The size of such a forest for two trees corresponds to their
  subtree prune-and-regraft distance and, if the agreement forest is acyclic,
  to their hybridization number.
  These distance measures are essential tools for understanding reticulate
  evolution.
  Our algorithm for computing maximum acyclic agreement forests is
  the first depth-bounded search algorithm for this problem.
  Our algorithms substantially outperform the best previous algorithms for
  these problems.
\end{abstract}

\begin{keywords}
  fixed-parameter tractability, phylogenetics, subtree prune-and-regraft
  distance, lateral gene transfer, hybridization, agreement forest.
\end{keywords}

\begin{AMS}
  68W05, 05C05, 05C76
\end{AMS}


\section{Introduction}

Phylogenetic trees are a standard model to represent
the evolutionary relationships among a set of species and are an indispensable
tool in evolutionary biology~\cite{hillis96}.
Early methods of building phylogenetic trees used morphology, or structural
characteristics of species, to determine their relatedness.
However, advances in molecular biology have allowed the widespread use of DNA
and protein sequence data to build phylogenies.
Molecular phylogenetics is particularly useful in the study of microscopic
organisms, due to their high rate of evolution and subtle
differences in appearance.
However, even good phylogenetic inference methods
cannot guarantee that a constructed tree correctly represents evolutionary
relationships---and there may not even exist such a tree---because
not all groups of species follow a simple
tree-like evolutionary pattern.
Collectively known as reticulation
events, nontree-like evolutionary processes, such as hybridization,
lateral gene transfer (LGT), and recombination, result in species being
composites of genes derived from different ancestors.
These processes allow species to rapidly acquire useful traits and
adapt to new environments.
This includes harmful
traits of pathogenic bacteria, such as antibiotic resistance, and LGT
appears to have contributed to the emergence of pathogens such as
\emph{Mycobacterium tuberculosis}~\cite{rosasmagallanes2006}.

Due to reticulation events, phylogenetic trees representing the
evolutionary history of different genes found in the same set of species
may differ.
To reconcile these differing evolutionary histories, one
can use phylogenetic distance measures that determine how well the
evolutionary hypotheses of two or more phylogenetic trees agree and
often allow us to discover reticulation events that explain the
differences.
To simultaneously represent these
discordant topologies, one can use a hybridization network,
which is a generalization of a phylogenetic tree that allows species
to inherit genetic material from more than one parent.

A number of distance measures are commonly used for comparing
phylogenies.
The Robinson-Foulds distance~\cite{robinson81} is
popular, as it can be calculated in linear time~\cite{day85}.
Other measures, such as the \emph{subtree prune-and-regraft} (SPR)
distance~\cite{hillis96} and the \emph{hybridization number}~\cite{baroni05},
are more biologically meaningful but are NP-hard to
compute~\cite{allen01,bordewich05,hickey2008sdc,bordewich07}.
The SPR distance is equivalent to the minimum number of lateral gene transfers
required to transform one tree into the
other~\cite{baroni05,beiko2006pil} and thus provides a lower bound on the
number of reticulation events needed to reconcile the two phylogenies.
The hybridization number of two phylogenies is the number of hybridization
events necessary to explain their differences.
These distance measures have been regularly used to model reticulate
evolution~\cite{maddison97,nakhleh05}, as
the minimum number of reticulation events required to reconcile two trees
provides the simplest explanation for the differences between the trees.
The close relationship between
SPR operations and reticulation events has also led to advances in
network models of evolution~\cite{baroni05,bordewich07,nakhleh05}.

Numerous researchers have focused and continue to focus on the development of
efficient algorithms to compute the distance between two trees using these
measures (see \S\ref{sec:related-work}).
In this paper, we present the currently fastest fixed-parameter
algorithms for computing the SPR distance and hybridization number
of two rooted binary phylogenies.
Our algorithms substantially outperform the best previous algorithms.
Similarly to previous algorithms for these problems, we
model these distance measures using
maximum agreement forests (MAFs) and maximum
acyclic agreement forests (MAAFs), respectively.
These are forests that can be obtained from either tree by cutting an
appropriate set of edges.
The edges that are not cut capture evolutionary relationships that
agree between both trees.
An agreement forest is maximal if it maximizes the number of these
agreeing relationships, that is, if it minimizes the number of edges that need
to be cut in either tree to obtain it.
Given an agreement forest obtained by removing $k$ edges from each tree, a set
of $k$ SPR operations that transform one tree into the other can
be recovered easily.
Similarly, if the agreement forest is \emph{acyclic}
(a restriction that disallows the donation of genetic information from
descendant nodes to ancestor nodes), a hybrid network with $k$
hybridization events can be constructed quickly~\cite{bordewich07}.
The core of the problem of computing the SPR distance or hybridization number
of two trees is thus to compute a maximum (acyclic) agreement forest.

\subsection{Related Work}

\label{sec:related-work}

While the SPR distance and hybridization number capture biologically meaningful
notions of similarity between phylogenies, their practical use has been limited
by the fact that they are NP-hard to
compute~\cite{allen01,bordewich05,hickey2008sdc,bordewich07}.
This has led to numerous efforts to develop approximation and fixed-paramater
algorithms, as well as heuristics, for computing these distances.

Hein et~al.~\cite{hein96} introduced the notion of a maximum agreement forest
and used it as the main tool underlying a proposed NP-hardness proof and
$3$-approximation algorithm for computing the SPR distance between
\emph{unrooted} phylogenies.
The central claim was that the number of components in an MAF of two phylogenies
is one more than the minimum number of SPR operations needed to transform one
into the other.
Unfortunately, there were subtle mistakes in the proofs, and Allen and
Steel~\cite{allen01} proved that the number of components in an MAF is in fact
one more than the closely related tree bisection and reconnection (TBR) distance
between the two trees.
Rodrigues et~al.~\cite{rodrigues2007maf} provided instances where the algorithm
of~\cite{hein96} provides an approximation guarantee no better than $4$ for the
size of an MAF, thereby disproving the $3$-approximation claim of~\cite{hein96}.
They also proposed a modification to the algorithm, which they claimed to
produce a $3$-approximation for the TBR distance.
A counterexample to this claim
was provided by Bonet et~al.~\cite{bonet06}, who showed, however, that
both the algorithms of~\cite{hein96} and~\cite{rodrigues2007maf}
compute $5$-approximations of the SPR distance between two \emph{rooted}
phylogenies, and that the algorithms can be implemented in linear time.
The approximation ratio was improved to 3 by Bordewich et
al.~\cite{bordewich08}, but at the expense of an increased running
time of $\OhL{n^5}$.\footnote{Using nontrivial but standard data
  structures, the running time can be reduced to $\OhL{n^4}$.}
  A second 3-approximation algorithm presented in~\cite{rodrigues2007maf}
achieves a running time of $\OhL{n^2}$.
Using entirely different ideas, Chataigner~\cite{chataigner05} obtained an
$8$-approximation algorithm for TBR distances of two or more trees.
There is currently no constant-factor approximation algorithm for the
hybridization number of two rooted phylogenies.
Kelk et~al.~\cite{kelk11} recently provided an explanation for the difficulty of
obtaining such an algorithm by proving that the hybridization number of two
phylogenetic trees and the size of a minimum feedback vertex set of a
directed graph are equally hard to approximate.

Given that the identification of meaningful putative reticulation events
from two phylogenetic trees is possible only if the trees carry a strong
vertical signal, that is, if the number of reticulation events is small compared
to the size of the trees, a promising approach to compute SPR distances and
hybridization numbers exactly is to use fixed-parameter algorithms that use
the distance $k$ between the two trees as parameter.
The previously best such algorithm for rooted SPR distance is due to
Bordewich et~al.~\cite{bordewich08} and runs in $\OhL{4^k \cdot k^4 + n^3}$
time.
For TBR distance, the best previous result is due to Hallett and
McCartin~\cite{hallett07}, who provided an algorithm with running time
$\OhL{4^k \cdot k^5 + p(n)}$, where $p(\cdot)$ is a polynomial function.
An earlier algorithm for this problem by Allen and Steel\cite{allen01} had
running time $\OhL{k^{3k} + p(n)}$.
For unrooted SPR, Hickey et~al.~\cite{hickey06tr} first claimed a
fixed-parameter algorithm, but the correctness proof was flawed.
Recently, Bonet and St.\ John~\cite{bonet-uspr} presented a corrected
proof that unrooted SPR is fixed-parameter tractable.
In \cite{bordewich07chn}, Bordewich and Semple provided a fixed-parameter
algorithm for the hybridization number of two rooted phylogenies
with running time $\OhL{(28k)^k + n^3}$.
Linz and Semple~\cite{linz09hnt} extended these results to nonbinary rooted
phylogenies.
Kelk et~al.~\cite{kelk11} provided an improved analysis of the kernel
size for hybridization number, which reduces the running time of the algorithm
by Bordewich and Semple to $\OhL{(18k)^k + n^3}$.
Chen and Wang~\cite{chen12} recently proposed an algorithm for computing all
MAAFs of two or more binary phylogenies.
Their algorithm combines the $\OhL{3^k n}$ search for agreement forests from
\cite{whidden2009uva} with an exhaustive search based on an observation in the
same paper that a superforest of an MAAF can be refined to an MAAF by cutting
appropriate edges incident to the roots in the current forest.

Numerous heuristic approaches for computing SPR distances have also been
proposed.
LatTrans by Hallet and Lagergen~\cite{hallett2001eal} models lateral
gene transfer events by a restricted version of rooted SPR operations,
considering two ways in which the trees can differ.
It computes the exact distance under this restricted metric in
$\OhL{2^kn^2}$ time.
HorizStory by Macleod et~al.~\cite{macleod2005dpe} supports multifurcating trees
but does not consider SPR operations where the pruned subtree contains more than
one leaf.
EEEP by Beiko and Hamilton~\cite{beiko2006pil} performs a breadth-first
SPR search on a rooted start tree but performs unrooted comparisons
between the explored trees and an unrooted reference tree.
The distance returned is not guaranteed to be exact, due to optimizations
and heuristics that limit the scope of the search, although EEEP provides
options to compute the exact unrooted SPR distance with no
nontrivial bound on the running time.
More recently, RiataHGT by Nakhleh et~al.~\cite{nakhleh2005rhf} computes an
approximation of the SPR distance between rooted multifurcating trees in
polynomial time.

Two algorithms for computing rooted SPR distances,
SPRdist~\cite{wu2009practical} and TreeSAT~\cite{bonet2009efficiently},
express the problem of computing maximum agreement forests as
an integer linear program (ILP) and a satisfiability problem (SAT),
respectively, and employ efficient ILP and SAT solvers to obtain a
solution.
SPRdist has been shown to outperform EEEP and Lattrans~\cite{wu2009practical}.
Although such algorithms draw on the close scrutiny that has been applied to
these problems, experiments show that these algorithms cannot compete with
the rooted SPR algorithm presented in this paper \cite{whidden2010fast}.

\subsection{Contribution}

\begin{table}[t]
  \footnotesize\centering
  \setlength{\tabcolsep}{0pt}
  \caption{Previous and new results on FPT algorithms for rooted SPR distance
    and hybridization number}
  \label{tbl:results}
  \begin{tabular}{l@{\hspace{1em}}l@{\hspace{1em}}l}
    \toprule
                                  & \textbf{Previous}                                      & \textbf{New}\\
    \midrule
    \textbf{Rooted SPR distance}  & $\OhL{4^k k^4 + n^3}$ time \cite{bordewich08}           & $\OhL{2.42^k k + n^3}$ or $\OhL{2.42^k n}$ time\\
    \textbf{Hybridization number} & $\OhL{(18k)^k + n^3}$ time \cite{bordewich07chn,kelk11} & $\OhL{3.18^k k + n^3}$ or $\OhL{3.18^k n}$ time\\
    \bottomrule
  \end{tabular}
\end{table}

Our contribution is to develop substantially more efficient algorithms for
computing the SPR distance and the hybridization number of two rooted binary
phylogenetic trees.
Using a ``shifting lemma'' central to Bordewich et~al.'s $3$-approximation
algorithm~\cite{bordewich08}, one can obtain a depth-bounded search algorithm
for computing the SPR distance with running time $\OhL{3^k n}$
\cite{whidden2009uva}.
We analyze the structure of rooted agreement forests further and identify
three distinct subcases that allow us to improve the algorithm's
running time to $\OhL{2.42^k n}$.
By combinining this result with kernelization rules by Bordewich and
Semple~\cite{bordewich05}, we obtain an algorithm with running time
$\OhL{2.42^k k + n^3}$.
Table~\ref{tbl:results}
shows our new results in comparison to the best previous results.
We note here that the approach discussed in this paper also leads to
linear-time $3$-approximation algorithms for rooted SPR distance and unrooted
TBR distance, as well as to an $\OhL{4^k k + n^3}$-time algorithm for unrooted
TBR distance.
Details can be found in~\cite{whidden2009MCS,whidden2009uva,whidden2011arxiv}.

In~\cite{whidden2009uva,whidden2010fast} we also claimed
results on computing MAAFs, but we used an incorrect definition of an
acyclic agreement forest that considers only cycles of length 2.
The algorithm consisted of two phases.
First we produce an agreement forest that is guaranteed to be a supergraph of
an MAAF.
Then we cut additional edges to eliminate cycles.
The first phase is not affected by our incorrect definition of cycles.
To implement the second phase correctly, we present a novel method in this paper
whose performance is close to the one claimed in~\cite{whidden2010fast}.
Obtaining this solution requires substantial new insights into the structure of
acyclic agreement forests beyond the results already published
in~\cite{whidden2009uva,whidden2010fast} and previous work.
Our algorithm is the first depth-bounded search algorithm for computing
hybridization numbers and substantially outperforms existing methods.

The rest of this paper is organized as follows.
In \S\ref{sec:prelim}, we introduce the necessary terminology and
notation.
in \S\ref{sec:fpt}, we present our algorithm for computing rooted MAFs.
In \S\ref{sec:refinement}, we present our MAAF algorithm.
This section consists of 5 parts, each of which presents one key tool.
We first develop a refined cycle graph,
analyze cycles in agreement forests, and identify subsets of edges that
can be removed from a cyclic agreement forest to give an MAAF.
These methods together provide a simple cycle breaking step that leads to an
MAAF algorithm with running time $\OhL{9.68^k n}$.
We then analyze the tree space explored by our depth-bounded search algorithm to
halve the exponential base in the running time of the cycle breaking algorithm
and thus obtain an MAAF algorithm with running time $\OhL{4.84^k n}$.
We conclude this section with an improved analysis, which shows that
only slight modifications to the cycle breaking procedure in the
$\OhL{4.84^k n}$-time algorithm lead to a greatly improved running time of
$\OhL{3.18^k n}$.
The $\OhL{3.18^k k + n^3}$-time algorithm in Table~\ref{tbl:results} is obtained
once again by combining our algorithm with known kernelization rules
\cite{bordewich07chn}.
In \S\ref{sec:concl}, we present concluding remarks and suggest future
work.

\section{Preliminaries}

\label{sec:prelim}

Throughout this paper, we mostly use the definitions and notation
from~\cite{allen01,bonet06,bordewich05,bordewich08,rodrigues2007maf}.
A \emph{(rooted binary phylogenetic) $X$-tree} is a rooted tree $T$
whose nodes each have zero or two children.
The leaves are bijectively labelled with the members of a label set $X$.
As in~\cite{bonet06,bordewich05,bordewich08,rodrigues2007maf}, we
augment the tree with a labelled root node whose label is distinct
from the labels of all leaves and whose only child is the original root
of $T$; see Figure~\ref{fig:x-tree}.
In the remainder of this paper, we consider $\rho$ to be part of $X$.
For a subset $V$ of $X$, $\subtree{V}$ is the
smallest subtree of $T$ that connects all nodes in~$V$; see
Figures~\ref{fig:subtree};
The \emph{$V$\!-tree induced by $T$} is the smallest tree $\induced{V}$ that
can be obtained from $\subtree{V}$ by suppressing unlabelled nodes with fewer
than two children; see Figure~\ref{fig:induced}.
\emph{Suppressing} a node $v$ deletes $v$ and its incident edges;
if $v$ is of degree $2$ with parent $u$ and child $w$, $u$ and $w$ are
reconnected using a new edge~$(u, w)$.

\begin{figure}[t]
  \footnotesize
  \hspace*{\stretch{1}}%
  \subfigure[\unskip\label{fig:x-tree}]{\includegraphics{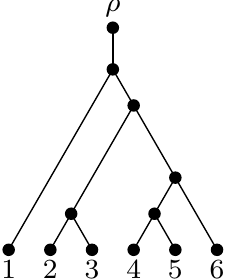}}%
  \hspace*{\stretch{1}}%
  \subfigure[\unskip\label{fig:subtree}]{\includegraphics{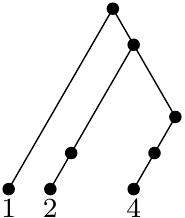}}%
  \hspace*{\stretch{1}}%
  \subfigure[\unskip\label{fig:induced}]{\includegraphics{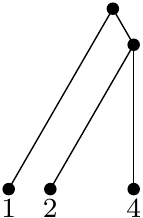}}%
  \hspace*{\stretch{1}}%
  \\
  \hspace*{\stretch{1}}%
  \subfigure[\unskip\label{fig:spr}]{\includegraphics{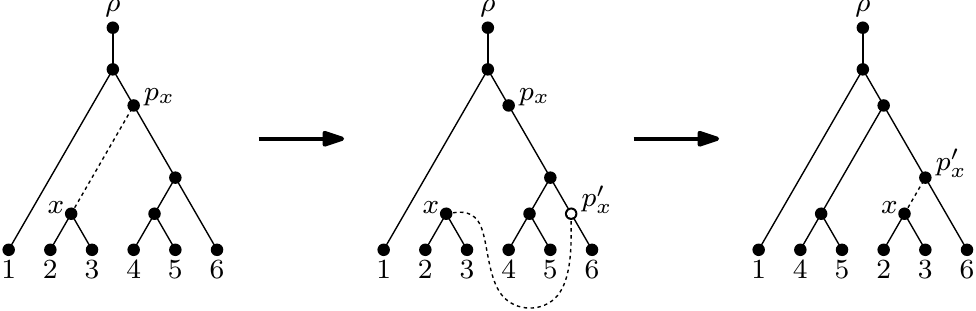}}%
  \hspace*{\stretch{1}}%
  \caption{(a)~An $X$-tree $T$.
    (b)~The subtree $\subtree{V}$ for $V = \set{1, 2, 4}$.
    (c)~$\induced{V}$.
    (d)~An SPR operation.}
\end{figure}

A \emph{subtree prune-and-regraft} (SPR) operation on an
$X$-tree $T$ cuts an edge $\edge{x} := (x, \parent{x})$, where $\parent{x}$
denotes the parent of $x$.
This divides $T$ into subtrees $T_x$ and $T_{\parent{x}}$ containing $x$
and~$\parent{x}$, respectively.
Then it introduces a new node $\parent{x}'$ into $T_{\parent{x}}$ by
subdividing an edge of $T_{\parent{x}}$ and adds an edge~$(x, \parent{x}')$,
thereby making $x$ a child of~$\parent{x}'$.
Finally, $\parent{x}$ is suppressed.
See Figure~\ref{fig:spr}.

SPR operations give rise to a distance measure $\dspr{\cdot, \cdot}$ between
$X$-trees, defined as the minimum number
of SPR operations required to transform one tree into the other.
The trees in Figure~\ref{fig:three-spr}, for example, have SPR distance
$\dspr{T_1, T_2} = 3$.

\begin{figure}[t]
  \footnotesize
  \hspace*{\stretch{1}}%
  \subfigure[\unskip\label{fig:three-spr}]{\includegraphics{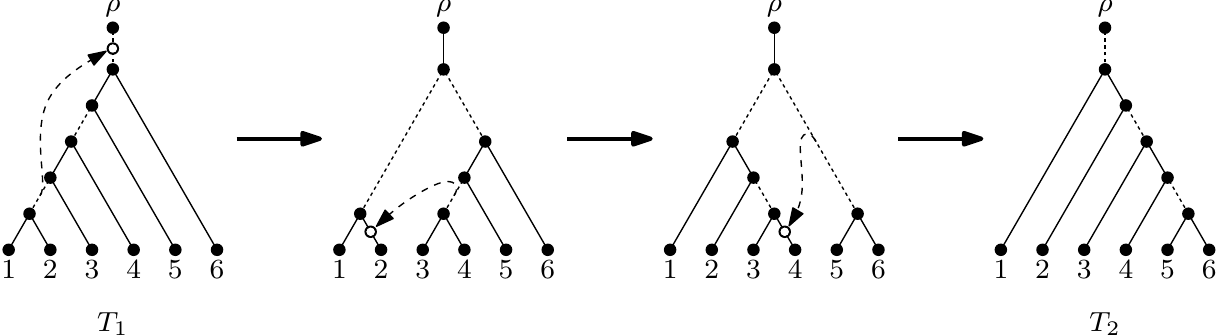}}%
  \hspace*{\stretch{1}}%
  \\
  \hspace*{\stretch{1}}%
  \subfigure[\unskip\label{fig:maf}]{\includegraphics{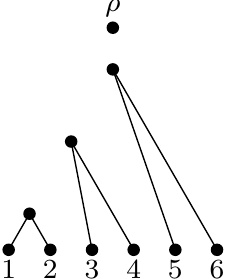}}
  \hspace*{\stretch{1}}%
  \subfigure[\unskip\label{fig:hybrid}]{\includegraphics{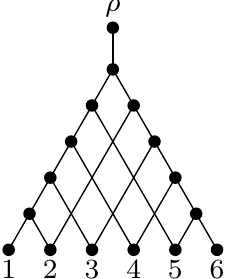}}%
  \hspace*{\stretch{1}}%
  \subfigure[\unskip\label{fig:maaf}]{\includegraphics{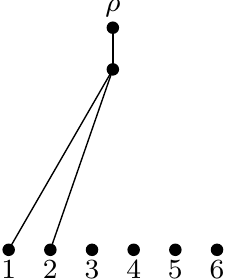}}%
  \hspace*{\stretch{1}}%
  \\
  \hspace*{\stretch{1}}%
  \subfigure[\unskip\label{fig:maaf:g_maf}]{\includegraphics{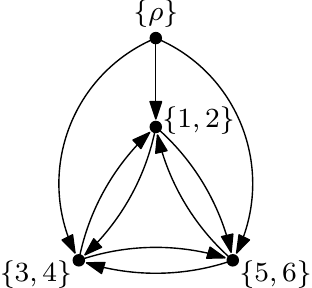}}%
  \hspace*{\stretch{1}}%
  \subfigure[\unskip\label{fig:maaf:g_maaf}]{\includegraphics{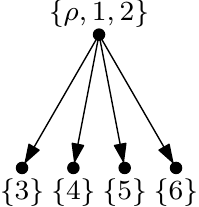}}
  \hspace*{\stretch{1}}%
  \caption{(a) SPR operations transforming $T_1$ into $T_2$.
    Each operation changes the top endpoint of one of the dotted edges.
    (b) The corresponding agreement forest, which can be obtained by
    cutting the dotted edges in both trees.
    This is an MAF with 4 components,
    so $\maf{T_1, T_2} = 4$ and $\ecut{T_1, T_2, T_2} = \dspr{T_1, T_2}$ = 3.
    Note that this is not an MAAF, as its cycle graph, shown in (e), contains
    a cycle.
    (c) A hybrid network of $T_1$ and $T_2$.
    This network has 4 nodes with an extra parent, so the hybridization number
    is 4.
    (d) An MAAF of $T_1$ and $T_2$.
    \mbox{$\aecut{T_1, T_2, T_2} = \dhyb{T_1, T_2} = 4$.}
    Note that this is not an MAF of $T_1$ and $T_2$, as it has one more
    component than the MAF in (b).
    (e)~The cycle graph of the agreement forest in (b), which
    contains a cycle.
    (f) The cycle graph of the agreement forest in (d), which
    does not contains a cycle.}
\end{figure}

A related distance measure for $X$-trees is their
\emph{hybridization number}, $\dhyb{T_1, T_2}$, which is defined in
terms of hybrid networks of the two trees.
A \emph{hybrid network} of two $X$-trees $T_1$ and $T_2$ is a directed acyclic
graph $H$ with a single source $\rho$, whose sinks are labelled bijectively with
the labels in $X \setminus \set{\rho}$, and such
that both $T_1$ and~$T_2$, with their edges directed away from the root, can be
obtained from $H$ by deleting edges and suppressing nodes.
For a vertex $x \in H$, let $\indeg{x}$ be its in-degree.
Then the hybridization number of $T_1$ and $T_2$ is
$\min_H \sum_{x \in H,x \ne \rho} (\indeg{x} - 1)$,
where the minimum is taken over all hybrid networks $H$ of $T_1$ and $T_2$.
This is illustrated in Figure~\ref{fig:hybrid}.

These distance measures are related to the sizes of
appropriately defined agreement forests.
To define these, we first introduce some terminology.
For a forest $F$ whose components are rooted phylogenetic trees
$T_1, T_2, \dots, T_k$ with label sets $X_1, X_2, \dots, X_k$,
we say $F$ \emph{yields} the forest with
components $\induced[T_1]{X_1}, \induced[T_2]{X_2}, \dots,
\induced[T_k]{X_k}$; if $X_i = \emptyset$, then $\subtree[T_i]{X_i} =
\emptyset$ and, hence, $\induced[T_i]{X_i} = \emptyset$.
In other words, the forest yielded by $F$ is the smallest forest that can
be obtained from $F$ by suppressing unlabelled nodes with less than two
children.
For a subset $E$ of edges of $F$, we use $F - E$ to denote the forest obtained
by deleting the edges in $E$ from $F$, and $F \div E$ to denote the
forest yielded by $F - E$.
We say $F \div E$ is a \emph{forest of~$F$}.

\begin{figure}[t]
  \centering
  \includegraphics{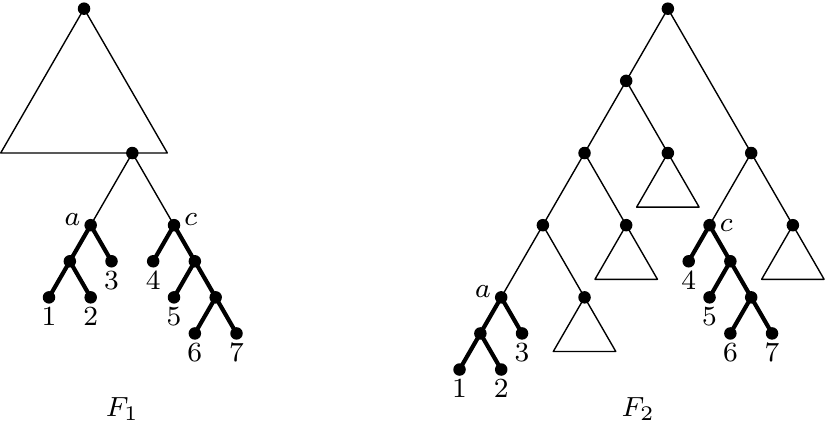}
  \caption{A sibling pair $(a,c)$ of two forests $F_1$ and $F_2$: $a$ and
  $c$ have a common parent in $F_1$, and both $a$ and $c$ exist also in $F_2$.}
	\label{fig:sibling-pair}
\end{figure}

Given $X$-trees $T_1$ and $T_2$ and forests $F_1$ of $T_1$ and $F_2$
of $T_2$, a forest $F$ is an \emph{agreement forest} (AF) of $F_1$ and
$F_2$ if it is a forest of both $F_1$ and $F_2$.
$F$ is a \emph{maximum agreement forest} (MAF) of $F_1$ and $F_2$ if there is
no AF of $F_1$ and $F_2$ with fewer components.
We denote the number of components in an MAF of $F_1$ and~$F_2$ by $\maf{F_1,
  F_2}$.
For a forest $F$ of $F_1$ or $F_2$, we use
$\ecut{F_1, F_2, F}$ to denote the size of the smallest edge set $E$ such that
$F \div E$ is an AF of $F_1$ and $F_2$.
Bordewich and Semple~\cite{bordewich05} showed that, for two
$X$-trees $T_1$ and $T_2$, $\dspr{T_1, T_2} = \ecut{T_1, T_2, T_2} =
\maf{T_1, T_2} - 1$.
An MAF of the trees in Figure~\ref{fig:three-spr} is shown in
Figure~\ref{fig:maf}.

The hybridization number of two $X$-trees $T_1$ and $T_2$ corresponds to an MAF
of $T_1$ and $T_2$ with an additional constraint.
For two forests $F_1$ and $F_2$ of $T_1$ and $T_2$ and an
AF $F = \{C_\rho, C_1, C_2, \ldots, C_k\}$ of $F_1$ and $F_2$, we
define a \emph{cycle graph} $G_F$ of~$F$.
Each node of $G_F$ represents a component of $F$, and there is an edge from
node $C_i$ to node $C_j$ if $C_i$ is an ancestor of $C_j$ in one of the
trees.
Formally, we map every node $x \in F$ to two nodes $\phi_1(x) \in T_1$
and $\phi_2(x) \in T_2$ by defining
$\phi_i(x)$ to be the lowest common ancestor in $T_i$ of all labelled leaves
that are descendants of $x$ in $F$.
We refer to $\phi_1(x)$ and $\phi_2(x)$ simply as $x$ in this paper, except when
this creates confusion.
For two components $C_i$ and $C_j$ of $F$ with roots $r_i$ and $r_j$,
$G_F$~contains the edge $(C_i, C_j)$ if and
only if either $\phi_1(r_i)$ is an ancestor of $\phi_1(r_j)$
or $\phi_2(r_i)$ is an ancestor of $\phi_2(r_j)$.
We say $F$ is \emph{cyclic} if $G_F$ contains a directed cycle.
Otherwise $F$ is an \emph{acyclic} agreement forest (AAF) of $F_1$ and $F_2$.
A \emph{maximum acyclic agreement forest} (MAAF) of $F_1$ and
$F_2$ is an AAF with the minimum number of components.
We denote its size by $\maaf{F_1, F_2}$ and the number of edges
in a forest $F$ of $F_1$ or $F_2$ that must be cut to obtain an AAF of $F_1$ and
$F_2$ by $\aecut{F_1, F_2, F}$.
Baroni et al.~\cite{baroni05} showed that
$\dhyb{T_1, T_2} = \aecut{T_1, T_2, T_2} = \maaf{T_1, T_2} - 1$.
An MAAF of the trees in Figure~\ref{fig:three-spr} is shown in
Figure~\ref{fig:maaf}.
The cycle graphs for the MAF and MAAF of these trees
shown in Figures~\ref{fig:maf} and~\ref{fig:maaf}
are shown in Figures~\ref{fig:maaf:g_maf} and~\ref{fig:maaf:g_maaf},
respectively.

For two nodes $a$ and $b$ of a forest $F$, we write $a \reach[F] b$ if there
exists a path between $a$ and $b$ in~$F$.
An \emph{internal node} of a path $P$ in $F$ is a node of $P$ that
is not an endpoint of $P$; a \emph{pendant node} of $P$ is a node not in $P$
and whose parent is an internal node of $P$.
For a node $x$ of a rooted forest $F$, $F^x$ denotes the subtree of $F$ induced
by all descendants of $x$, including $x$.
For two rooted forests $F_1$ and $F_2$ and a node $a \in F_1$, we say that $a$
\emph{exists} in $F_2$ if there is a node $a' \in F_2$ such that
$F_1^a = F_2^{a'}$.
For simplicity, we refer to both $a$ and $a'$ as $a$.
For forests $F_1$ and $F_2$ and nodes $a, c \in F_1$ with a common parent, we
say $(a,c)$ is a \emph{sibling pair} of $F_1$ if $a$ and $c$ exist in $F_2$.
Figure~\ref{fig:sibling-pair} shows such a sibling pair.

The correctness proofs of our algorithms in the next sections make use
of the following two lemmas.
Lemma~\ref{lem:edge-shift} was shown
by Bordewich et al.~\cite{bordewich08} and is illustrated in
Figure~\ref{fig:edge-shift}.
Suppose we cut a set of edges $E$
from a forest $F$ to obtain $F \div E$, and there is an edge $e$ of $F$
such that $F - (E \cup \set{e})$ has a component without labelled nodes.
This lemma shows that the forest
$F \div (E \setminus \set{f} \cup \set{e})$
obtained by replacing any edge $f \in E$ on the boundary of this
``empty'' component with $e$ is the same as~$F \div E$.

\begin{figure}[bt]
  \subfigure[\unskip\label{fig:isomorphic}]{\includegraphics{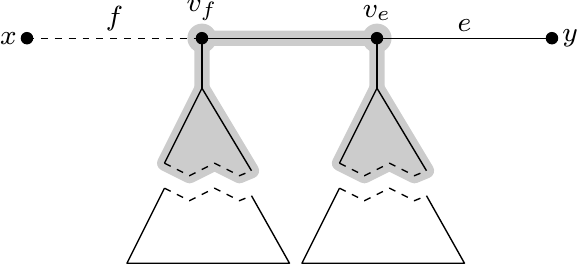}}%
  \hspace{\stretch{1}}%
  \subfigure[\unskip\label{fig:nisomorphic}]{\includegraphics{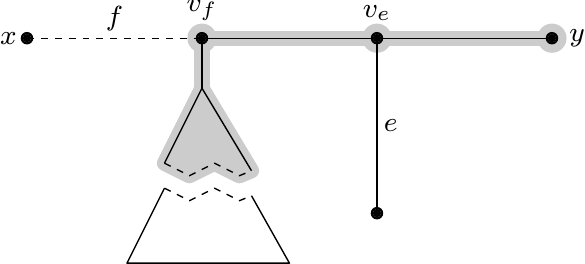}}%
  \caption{Illustration of the Shifting Lemma.
    (a) The lemma applies because $e$ and $f$ are on the boundary of an
    ``empty'' component of $F - (E \cup \set{e})$, shown in grey.
    (b) The lemma does not apply because the component with $e$ and $f$
    on its boundary contains a labelled leaf $y$:
    $v_f \reach[F - (E \cup \set{e})] y$.}
  \label{fig:edge-shift}
\end{figure}

\begin{lemma}[Shifting Lemma]
  \label{lem:edge-shift}
  Let $F$ be a forest of an $X$-tree, $e$ and $f$ edges of~$F$, and
  $E$ a subset of edges of $F$ such that $f \in E$ and $e \notin E$.
  Let $v_f$ be the end vertex of $f$ closest to $e$, and $v_e$ an end
  vertex of $e$.
  If $v_f \reach[F-E] v_e$ and $x \noreach[F - (E \cup
  \set{e})] v_f$, for all $x \in X$, then $F \div E = F \div (E
  \setminus \set{f} \cup \set{e})$.
\end{lemma}

Let $F_1$ and $F_2$ be forests of $X$-trees $T_1$ and $T_2$, respectively.
Any agreement forest of $F_1$ and $F_2$ is an agreement forest of $T_1$ and
$T_2$.
Conversely, an agreement forest of $T_1$ and $T_2$ is an agreement forest of
$F_1$ and $F_2$ if it is a forest of $F_2$
and there are no two leaves $a$ and $b$ such that $a \reach[F_2]
b$ but $a \noreach[F_1] b$.
This is formalized in the following lemma.
Our algorithms ensure that any intermediate forests
$F_1$ and $F_2$ they produce have this latter property.
Thus, we can reason about agreement forests of $F_1$
and $F_2$ and of $T_1$ and $T_2$ interchangeably.

\begin{lemma}
  \label{lem:forest_ecut}
  Let $F_1$ and $F_2$ be forests of $X$-trees $T_1$ and $T_2$, respectively.
  Let $F_1$ be the union of trees $\dot{T}_1,
  \dot{T}_2, \ldots, \dot{T}_k$ and $F_2$ be the union of forests
  $\dot{F}_1, \dot{F}_2, \ldots, \dot{F}_k$ such that $\dot{T}_i$ and
  $\dot{F}_i$ have the same label set, for all $1 \le i \le k$.
  A forest of $F_2$ is an AF of
  $T_1$ and $T_2$ if and only if it is an AF of $F_1$ and $F_2$.
\end{lemma}

A \emph{triple} $\triple{ab}{c}$ of a rooted forest $F$ is defined by a set
$\set{a, b, c}$ of three leaves in the same component of $F$ and such
that the path from $a$ to $b$ in $F$ is disjoint from the path from
$c$ to the root of the component.
A triple of a forest $F_1$ is
\emph{compatible} with a forest $F_2$ if it is also a triple of $F_2$;
otherwise it is \emph{incompatible} with $F_2$.
An agreement forest of two forests $F_1$ and $F_2$ cannot contain a
triple incompatible with either of the two forests.
Thus, we have the following observation.

\begin{observation}
  \label{obs:incompatible-triple}
  Let $F_1$ and $F_2$ be forests of rooted $X$-trees $T_1$ and $T_2$,
  and let $F$ be an agreement forest of $F_1$ and $F_2$.
  If $\triple{ab}{c}$ is a
  triple of $F_1$ incompatible with $F_2$, then $a \noreach[F] b$ or
  $a \noreach[F] c$.
\end{observation}

For two forests $F_1$
and $F_2$ with the same label set, two components $C_1$ and $C_2$ of
$F_1$ are said to \emph{overlap} in $F_2$ if there exist leaves $a, b
\in C_1$ and $c, d \in C_2$ such that the paths from $a$ to $b$ and
from $c$ to $d$ in $F_2$ exist and are nondisjoint.
Since we consider only binary trees in this paper, this means
the two paths share an edge.
The following lemma is an easy extension of a lemma of \cite{bordewich08}, which
states the same result for a tree $T_2$ instead of a forest $F_2$.

\begin{lemma}
  \label{lem:forest-condition}
  Let $F_1$ and $F_2$ be forests of two $X$-trees $T_1$ and $T_2$, and
  denote the label sets of the components of $F_1$ by $X_1, X_2,
  \dots, X_k$ and the label sets of the components of $F_2$ by $Y_1,
  Y_2, \dots, Y_l$.
  $F_2$~is a forest of $F_1$ if and only if (1) for
  every $Y_j$, there exists an $X_i$ such that $Y_j \subseteq X_i$, (2)
  no two components of $F_2$ overlap in $F_1$, and (3) no triple of
  $F_2$ is incompatible with $F_1$.
\end{lemma}

\section{Computing the SPR Distance}

\label{sec:fpt}

In this section, we present our algorithm for computing the SPR distance
of two $X$-trees.
It will be obvious from the description of the algorithm that it also produces a
corresponding MAF.
We do not discuss this further in the remainder of this section and focus only
on computing $\dspr{T_1, T_2}$.

As is customary for FPT algorithms, we focus on the decision version of the
problem: ``Given two $X$-trees $T_1$ and $T_2$ and a parameter $k$,
is $\dspr{T_1, T_2} \le k$?''
To compute the distance between two trees, we start with $k = 0$ and increase it
until we receive an affirmative answer.
This does not increase the running time of the algorithm by more than a constant
factor, as the running time depends exponentially on $k$.
The following theorem states the main result of this section.

\begin{theorem}
  \label{thm:maf:fpt}
  For two rooted $X$-trees $T_1$ and $T_2$ and a parameter $k$, it
  takes $\OhXL{\parensL{1+\sqrt{2}}^k n} = \OhL{2.42^k n}$ time to decide
  whether $\ecut{T_1, T_2, T_2} \le k$.
\end{theorem}

Using reduction rules by Bordewich et al.~\cite{bordewich05}, we can improve
the running time in Theorem~\ref{thm:maf:fpt} for values of $k$ such that
$k \ge 2\log_{2.42} n$ and $k = \oh{n}$.
Given two trees $T_1$ and $T_2$,
these reduction rules take $\OhL{n^3}$ time to produce two trees
$T_1'$ and $T_2'$ of size at most $c \cdot \ecut{T_1, T_2, T_2}$ each,
for some constant $c > 0$ (determined by Bordewich et al.),
and such that $\ecut{T_1', T_2', T_2'} = \ecut{T_1, T_2, T_2}$.
If one of the trees has size greater than~$ck$, then $\ecut{T_1, T_2, T_2} > k$,
and we can answer ``no'' without any further processing.
If both trees have size at most $ck$, we can apply Theorem~\ref{thm:maf:fpt}
to $T_1'$ and $T_2'$ to decide in $\OhL{2.42^k k}$ time
whether $\ecut{T_1', T_2', T_2'} \le k$.
Thus, we obtain the following corollary.

\begin{corollary}
  \label{cor:maf}
  For two rooted $X$-trees $T_1$ and $T_2$ and a parameter $k$, it
  takes $\OhL{2.42^k k + n^3}$ time to decide whether $\ecut{T_1, T_2,
    T_2} \le k$.
\end{corollary}

In the remainder of this section, we prove Theorem~\ref{thm:maf:fpt}.
Our algorithm is recursive.
Each invocation takes two forests $F_1$ and $F_2$ of $T_1$ and $T_2$
and a parameter~$k$ as inputs, and decides whether $\ecut{T_1, T_2, F_2} \le k$.
We denote such an invocation by $\alg{F_1, F_2, k}$.
The forest $F_1$ is the union of a tree $\dot{T}_1$ and a forest~$F$
disjoint from~$\dot{T}_1$, while $F_2$ is the union of the same forest $F$ and
another forest $\dot{F}_2$ with the same label set as $\dot{T}_1$.
We maintain two sets of labelled nodes: $R_d$ (roots-done) contains the roots of
$F$, and $R_t$ (roots-todo) contains roots of (not necessarily maximal) subtrees
that agree between $\dot{T}_1$ and $\dot{F}_2$.
We refer to the nodes in these sets by their labels.
For the top-level invocation, $F_1 = \dot{T}_1 = T_1$, $F_2 =
\dot{F}_2 = T_2$, and $F = \emptyset$; $R_d$ is empty, and $R_t$
contains all leaves of $T_1$.

$\alg{F_1, F_2, k}$ identifies a small collection
$\set{E_1, E_2, \dots, E_q}$ of subsets of edges of $\dot{F}_2$ such
that $\ecut{T_1, T_2, F_2} \le k$ if and only if $\ecut{T_1, T_2, F_2
  \div E_i} \le k - \size{E_i}$, for at least one $1 \le i \le q$.
It makes a recursive call $\alg{F_1, F_2 \div E_i, k - \size{E_i}}$, for
each subset $E_i$, and returns ``yes'' if and only if one of these
calls does.
The steps of this procedure are as follows.

\begin{enumerate}[leftmargin=*]
\item \label{case:abort} (Failure) If $k < 0$, there is no subset $E$
  of at most $k$ edges of $F_2$ such that $F_2 - E$ yields an AF of
  $T_1$ and $T_2$: $\ecut{T_1, T_2, F_2} \ge 0 > k$.
  Return ``no'' in this case.
\item \label{case:success} (Success) If $|R_t| \le 2$, then $\dot{F}_2
  \subseteq \dot{T}_1$.
  Hence, $F_2 = \dot{F}_2 \cup F$ is an AF of $F_1$ and $F_2$ and, by
  Lemma~\ref{lem:forest_ecut}, also of $T_1$ and $T_2$.
  Thus, $\ecut{T_1, T_2, F_2} = 0 \le k$.
  Return ``yes'' in this case.
\item \label{case:spr:singleton} (Prune maximal agreeing subtrees) If
  there is a node $r \in R_t$ that is a root in $\dot{F}_2$, remove $r$ from
  $R_t$ and add it to $R_d$, thereby moving the corresponding subtree of
  $\dot{F}_2$ to $F$; cut the edge $\edge{r}$ in $\dot{T}_1$ and
  suppress $r$'s parent in $\dot{T}_1$; return to Step~\ref{case:success}.
  This does not alter $F_2$ and, thus, neither $\ecut{T_1, T_2, F_2}$.
  If no such root $r$ exists, proceed to Step~\ref{item:choose-sib-pair}.
\item \label{item:choose-sib-pair}Choose a sibling pair $(a,c)$ in
  $\dot{T}_1$ such that $a, c \in R_t$.
\item \label{case:spr:sibling} (Grow agreeing subtrees) If $(a,c)$ is
  a sibling pair of $\dot{F}_2$, remove $a$ and $c$ from
  $R_t$; label their parent in both forests with $(a,c)$ and add it to
  $R_t$; return to Step~\ref{case:success}.
  If $(a,c)$ is
  not a sibling pair of $\dot{F}_2$, proceed to
  Step~\ref{case:spr:non-sibling}.
\item \label{case:spr:non-sibling} (Cut edges) Distinguish three
  cases (see Figure~\ref{fig:spr:cases}):
  \begin{enumerate}[label=\theenumi.\arabic{*}.,leftmargin=*,ref=\theenumi.\arabic{*}]
  \item \label{case:spr:sc} If $a \noreach[F_2] c$,
      call $\alg{F_1, F_2 \div \set{\edge{a}}, k-1}$
      and $\alg{F_1, F_2 \div \set{\edge{c}}, k-1}$ recursively.
  \item \label{case:spr:cob} If $a \reach[F_2] c$ and
    the path from $a$ to $c$ in $\dot{F}_2$ has only one pendant node $b$,
    call $\alg{F_1, F_2 \div \set{\edge{b}}, k-1}$ recursively.
  \item \label{case:spr:cab} If $a \reach[F_2] c$ and
      the path from $a$ to $c$ in $\dot{F}_2$ has $q \ge 2$ pendant nodes
      $b_1, b_2, \ldots, b_q$, call $\alg{F_1, F_2
	\div \set{\edge{b_1}, \edge{b_2}, \ldots, \edge{b_q}}, k-q}$,
      $\alg{F_1, F_2 \div \set{\edge{a}}, k-1}$, and
      $\alg{F_1, F_2 \div \set{\edge{c}}, k-1}$ recursively.
  \end{enumerate}
  Return ``yes'' if one of the recursive calls does; otherwise return
  ``no''.
\end{enumerate}

\begin{figure}[t]
  \centering
  \includegraphics{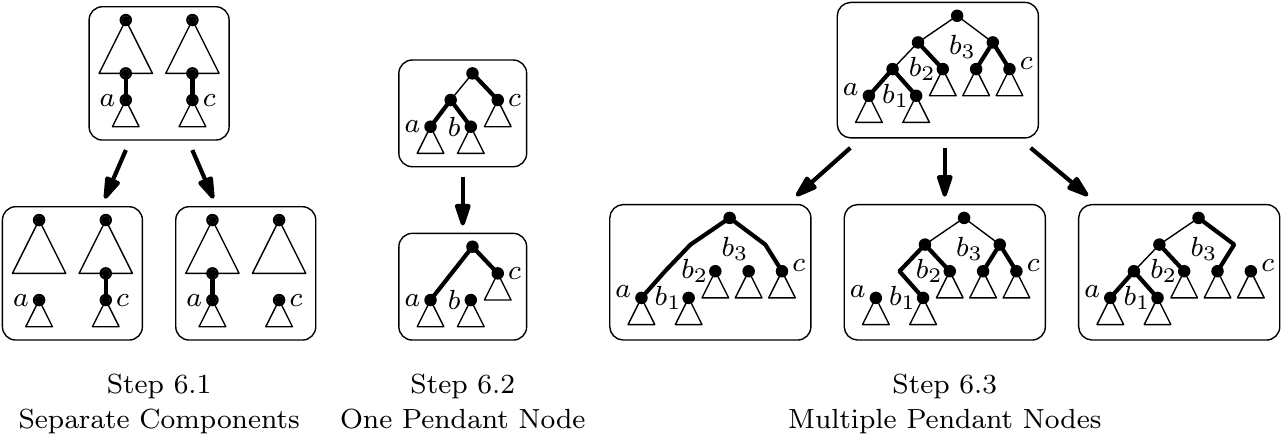}%
  \caption{The cases in Step~\ref{case:spr:non-sibling} of the
    MAF algorithm.
    Only $\dot{F}_2$ is shown.
    Each box represents a recursive call.}
  \label{fig:spr:cases}
\end{figure}

To prove that the algorithm achieves the running time stated in
Theorem~\ref{thm:maf:fpt}, we show that each invocation takes linear
time (Lemma~\ref{lem:maf:time}) and that the algorithm makes
$\OhXL{\parensL{1 + \sqrt{2}}^k}$ recursive calls (Lemma~\ref{lem:invocations}).

\begin{lemma}
  \label{lem:maf:time}
  Each invocation $\alg{F_1, F_2, k}$, excluding recursive calls it
  makes, takes linear time.
\end{lemma}

\begin{proof}
  We represent each forest as a collection of nodes, each of which points
  to its parent, left child, and right child.
  In addition, every labelled node (i.e., each node in $R_t$ or $R_d$) stores a
  pointer to its counterpart in the other forest.
  For $\dot{T}_1$, we maintain a list of sibling pairs of labelled nodes.
  Every labelled node of $\dot{T}_1$ stores a pointer to the pair it
  belongs to, if any.
  For $\dot{F}_2$, we maintain a list $R'_d \subseteq R_t$ of nodes that
  are roots of $\dot{F}_2$.
  This list is used to move these roots from $R_t$ to~$R_d$ in
  Step~\ref{case:spr:singleton}.

  It is easily verified that, using this representation of $F_1$ and $F_2$,
  each execution of Steps~1--5 takes constant time and that Step~6, excluding
  recursive calls it spawns, takes linear time.
  Steps~\ref{case:abort} and~\ref{case:spr:non-sibling} are executed
  only once per invocation.
  Steps~\ref{case:success}--\ref{case:spr:sibling} form a loop, and
  each iteration, except the first one, is the result of finding a root
  of $\dot{F}_2$ in Step~\ref{case:spr:singleton} or merging
  a sibling pair in Step~\ref{case:spr:sibling}.
  In the former case, Step~\ref{case:spr:singleton} cuts an edge in~$F_1$,
  which can happen only $\Oh{n}$ times because $F_1$ has $\Oh{n}$ edges.
  In the latter case, the number of nodes in $R_t$ decreases by one,
  which cannot happen more than $n$ times because the algorithm
  starts with the $n$ leaves of $T_1$ in $R_t$ and the number of
  nodes in $R_t$ never increases.
  Thus, Steps~\ref{case:success}--\ref{case:spr:sibling} are executed
  $\Oh{n}$ times, and the cost of the entire invocation is linear.\qquad
\end{proof}

\begin{lemma}
  \label{lem:invocations}
  An invocation $\alg{F_1, F_2, k}$ spawns $\OhXL{\parensL{1 + \sqrt{2}}^k}$
  recursive calls.
\end{lemma}

\begin{proof}
  Let $I(k)$ be the number of recursive calls spawned by an invocation
  with parameter $k$.
  By inspecting the different cases of Step~\ref{case:spr:non-sibling}, we
  obtain
\begingroup\allowdisplaybreaks
\begin{align*}
  I(k) &=
  \begin{cases}
    1                    & \text{if Step~\ref{case:spr:non-sibling} is not executed}\\
    1 + 2I(k-1)          & \text{Case~\ref{case:spr:sc}}\\
    1 + I(k-1)           & \text{Case~\ref{case:spr:cob}}\\
    1 + 2I(k-1) + I(k-q) & \text{Case~\ref{case:spr:cab}}
  \end{cases}\\
  &\le 1 + 2I(k-1) + I(k-2)
\end{align*}\endgroup
because Case~\ref{case:spr:cab} dominates the other two cases and $q
\ge 2$ in this case.
Simple substitution shows that this recurrence
solves to $I(k) = \OhXL{\parensL{1 + \sqrt{2}}^k}$.\qquad
\end{proof}

It remains to prove the
correctness of the algorithm, which we do by induction on $k$.
An invocation $\alg{F_1, F_2, k}$ with $k < 0$ correctly returns ``no''
in Step~\ref{case:abort}, so assume $k \ge 0$.
In this case, the invocation produces its answer in Step~\ref{case:success}
or~\ref{case:spr:non-sibling}.
If it produces its answer (``yes'') in Step~\ref{case:success}, this is correct
because $F_2$ is an MAF of $T_1$ and $T_2$.
If it produces its answer in Step~\ref{case:spr:non-sibling}, it suffices to
prove that $\ecut{T_1, T_2, F_2} \le k$ if and only if
$\ecut{T_1, T_2, F_2 \div E_i} \le k - \size{E_i}$, for at least one of the
recursive calls $\alg{F_1, F_2 \div E_i, k - \size{E_i}}$ the invocation makes
in Step~\ref{case:spr:non-sibling}.
This in turn follows if $\ecut{T_1, T_2, F_2 \div E_i} \ge \ecut{T_1, T_2, F_2}
- \size{E_i}$, for all recursive calls
$\alg{F_1, F_2 \div E_i, k - \size{E_i}}$, which is trivial,
and $\ecut{T_1, T_2, F_2 \div E_i} = \ecut{T_1, T_2, F_2} - \size{E_i}$,
for at least one recursive call $\alg{F_1, F_2 \div E_i, k - \size{E_i}}$.
Lemmas~\ref{lem:spr:sc}, \ref{lem:spr:cob}, and~\ref{lem:spr:cab} below
prove the latter for each case of Step~\ref{case:spr:non-sibling}.
For Cases~\ref{case:spr:sc} and~\ref{case:spr:cab}, we prove also that
$\aecut{T_1, T_2, F_2 \div E_i} = \aecut{T_1, T_2, F_2} -
\size{E_i}$, for at least one recursive call
$\alg{F_1, F_2 \div E_i, k - \size{E_i}}$.
This will be used in the correctness proof of the MAAF algorithm
in \S\ref{sec:refinement}.

In Step~\ref{case:spr:non-sibling}, $(a,c)$ is a sibling pair of $\dot{T}_1$
but not of $F_2$---otherwise Step~\ref{case:spr:sibling} would have replaced
$a$ and $c$ with their parent in $R_t$---and neither $F_2^a$ nor $F_2^c$ is a
component of~$F_2$---otherwise Step~\ref{case:spr:singleton} would have
removed $a$ or $c$ from $R_t$.
Note that $a$ and $c$ belong to~$\dot{F}_2$ because $\dot{T}_1$ and $\dot{F}_2$
have the same label set.
Let $b$ be $a$'s sibling in $F_2$.
If $a$ and $c$ belong to the same component of $F_2$, we assume w.l.o.g.\ that
$a$'s distance from the root of this component is no less than $c$'s.
Since $a$ and $c$ are not siblings in $F_2$, this implies that $c \notin F_2^b$.
If $a \noreach[F_2] c$, we also have $c \notin F_2^b$ because $a \reach[F_2] b$.

Our first lemma shows that we can always cut one of $\edge{a}$, $\edge{b}$, and
$\edge{c}$ to make progress towards an MAF or MAAF of $T_1$ and $T_2$ in
Step~\ref{case:spr:non-sibling}.
In \cite{whidden2009uva}, we used this as a basis for a simple $\OhL{3^kn}$-time
MAF algorithm.
Here, we need this lemma as a basis for the proofs of
Lemmas~\ref{lem:spr:sc}, \ref{lem:spr:cob}, \ref{lem:spr:cab},
and~\ref{lem:hyb:cob}.

\begin{lemma}
  \label{lem:spr:3-way}
  If $(a,c)$ is a sibling pair of $F_1$ and (i) $a \noreach[F_2] c$ and neither
  $F_2^a$ nor $F_2^c$ is a component of $F_2$ or (ii) $a \reach[F_2] c$ but $a$
  and $c$ are not siblings in $F_2$, then there exists an edge set $E$ of size
  $\ecut{T_1, T_2, F_2}$ (resp.\ $\aecut{T_1, T_2, F_2}$) and such that
  $F_2 \div E$ is an AF (resp.\ AAF) of $T_1$ and $T_2$ and $E \cap
  \set{\edge{a}, \edge{b}, \edge{c}} \ne \emptyset$.
\end{lemma}

\begin{proof}
  Consider an edge set $E$ of size $\ecut{T_1, T_2, F_2}$ and such that
  $F_2 \div E$ is an AF of $F_1$ and $F_2$, and assume $E$ contains the maximum
  number of edges from $\set{\edge{a}, \edge{b}, \edge{c}}$ among all edge sets
  satisfying these conditions.
  Assume for the sake of contradiction that $E \cap \set{\edge{a}, \edge{b},
    \edge{c}} = \emptyset$.

  If $a' \noreach[F_2 - E] a$, for all leaves $a' \in F_2^a$, then we choose
  an arbitrary such leaf $a' \in F_2^a$ and the first edge $f$ on the path
  from $a$ to $a'$.
  Lemma~\ref{lem:edge-shift} now implies that $F_2 - E$ and
  $F_2 - (E \setminus \set{f} \cup \set{\edge{a}})$ yield the same forest,
  which contradicts our choice of $E$.
  The same argument leads to a contradiction if $b' \noreach[F_2 - E] b$,
  for all leaves $b' \in F_2^b$, or $c' \noreach[F_2 - E] c$, for all leaves
  $c' \in F_2^c$.
  Thus, there exist leaves $a' \in F_2^a$, $b' \in F_2^b$, and $c' \in F_2^c$
  such that $a' \reach[F_2 - E] a$, $b' \reach[F_2 - E] b$, and
  $c' \reach[F_2 - E] c$.

  Since $(a,c)$ is a sibling pair of $\dot{T}_1$, $\triple{a'c'}{b'}$ is a
  triple of $F_1$, while $c \notin F_2^b$ implies that either
  $\triple{a'b'}{c'}$ is a triple of $F_2$ or $a' \noreach[F_2] c'$.
  In either case, the triple $\triple{a'c'}{b'}$ is incompatible with~$F_2$ and,
  by Observation~\ref{obs:incompatible-triple} and because $a' \reach[F_2 - E]
  b'$, we have $a' \noreach[F_2 - E] c'$ and, hence, $a'' \noreach[F_2 - E] c'$,
  for every leaf $a'' \in F_2^a$.
  Now, if there existed a leaf $x \notin F_2^c$ such that
  $c' \reach[F_2 - E] x$, then the components of $F_2 \div E$ containing $a'$
  and $c'$ would overlap in $F_1$: they would both include $\edge{p_a}$ because
  $b', x \notin F_1^{p_a}$.
  By Lemma~\ref{lem:forest-condition}, this would contradict that $F_2 \div E$
  is an AF of $F_1$ and $F_2$.
  Thus, no such leaf $x$ exists.
  On the other hand, since $F_2^c$ is not a component of $F_2$, there exists a
  leaf $x \notin F_2^c$ such that $c' \reach[F_2] x$.
  Since $x \noreach[F_2 - E] c'$, at least one edge on the path from $c'$ to
  $x$ belongs to $E$.
  Let $f$ be the first such edge.
  Since $c \reach[F_2 - E] c'$, $f$ does not belong to $F_2^c$.
  Hence, edges $\edge{c}$ and $f$ satisfy the conditions of
  Lemma~\ref{lem:edge-shift}, and $F_2 - E$ and $F_2 - (E \setminus
  \set{f} \cup \set{\edge{c}})$ yield the same forest, contradicting the choice
  of~$E$.

  The second claim of the lemma follows using the same arguments after choosing
  $E$ of size $\aecut{T_1, T_2, F_2}$ and such that $F \div E$ is an AAF of
  $T_1$ and $T_2$.\qquad
\end{proof}

The last three lemmas of this section now establish the correctness of each case
in Step~\ref{case:spr:non-sibling} of the algorithm and conclude the proof of
Theorem~\ref{thm:maf:fpt}.

\begin{lemma}[Case~\ref{case:spr:sc}---Separate Components]
  \label{lem:spr:sc}
  If $(a,c)$ is a sibling pair of $F_1$, $a \noreach[F_2] c$, and neither
  $F_2^a$ nor $F_2^c$ is a component of $F_2$, then
  there exists an edge set $E$ of size
  $\ecut{T_1, T_2, F_2}$ (resp.\ $\aecut{T_1, T_2, F_2}$) and such that
  $F_2 \div E$ is an AF (resp.\ AAF) of $T_1$ and $T_2$ and $E \cap
  \set{\edge{a}, \edge{c}} \ne \emptyset$.
\end{lemma}

\begin{proof}
  Consider an edge set $E$ of size $\ecut{T_1, T_2, F_2}$ and such
  that $F_2 \div E$ is an AF of $F_1$ and $F_2$, and assume $E$
  contains the maximum number of edges from $\set{\edge{a}, \edge{c}}$
  among all edge sets satisfying these conditions.
  Assume for the sake of contradiction that
  $E \cap \set{\edge{a}, \edge{c}} = \emptyset$.

  By the arguments in the proof of Lemma~\ref{lem:spr:3-way},
  there exist leaves $a' \in F_2^a$ and $c' \in F_2^c$ such that
  $a' \reach[F_2 - E] a$ and $c' \reach[F_2 - E] c$.
  Since $(a,c)$ is a
  sibling pair of $F_1$ but $a \noreach[F_2] c$ and, hence,
  $a' \noreach[F_2 - E] c'$, we must have $a' \noreach[F_2 - E'] x$, for every
  lea $x \notin F_2^a$, or $c' \noreach[F_2 - E] x$,
  for every leaf $x \notin F_2^c$.
  W.l.o.g.\ assume the latter.
  As shown in the proof of Lemma~\ref{lem:spr:3-way}, this implies that
  $F_2 - E$ and $F_2 - (E \setminus \set{f} \cup \set{\edge{c}})$ yield the same
  forest, where $f$ is the first edge on the path from $c'$ to a leaf
  $x \notin F_2^c$ and such that $c' \reach[F_2] x$.
  This contradicts the choice of $E$.

  The second claim of the lemma follows using the same arguments after choosing
  $E$ of size $\aecut{T_1, T_2, F_2}$ and such that $F \div E$ is an AAF of
  $T_1$ and $T_2$.\qquad
\end{proof}

\begin{lemma}[Case~\ref{case:spr:cob}---One Pendant Node---MAF]
  \label{lem:spr:cob}
  If $(a, c)$ is a sibling pair of~$F_1$, $a \reach[F_2] c$, and the path
  from $a$ to $c$ in $F_2$ has only one pendant node $b$, then there exists an
  edge set $E$ of size $\ecut{T_1, T_2, F_2}$ and such that $F_2 \div E$ is an
  AF of $T_1$ and $T_2$ and $\edge{b} \in E$.
\end{lemma}

\begin{proof}
  Again, consider an edge set $E$ of size $\ecut{T_1, T_2, F_2}$ and
  such that $F_2 \div E$ is an AF of $F_1$ and~$F_2$, and assume
  $E$ contains the maximum number of edges
  from $\set{\edge{a}, \edge{b}, \edge{c}}$ among all edge sets
  satisfying these conditions.
  By Lemma~\ref{lem:spr:3-way},
  $E \cap \set{\edge{a}, \edge{b}, \edge{c}} \ne \emptyset$.
  If $\edge{b} \in E$, there is nothing to prove, so assume
  $\edge{b} \notin E$.
  Let $v = \parent{a} = \parent{b}$, and $u = \parent{v} = \parent{c}$.

  If $F_2 \div (E \setminus \set{\edge{a}, \edge{c}, \edge{v}} \cup
  \set{\edge{b}})$ is an AF of $F_1$ and $F_2$, we are done because
  $E \cap \set{\edge{a}, \edge{c}} \ne \emptyset$ and, hence,
  $\size{E \setminus \set{\edge{a}, \edge{c}, \edge{v}} \cup \set{\edge{b}}}
  \le \size{E} = \ecut{T_1, T_2, F_2}$.
  So assume $F_2 \div (E \setminus \set{\edge{a}, \edge{c}, \edge{v}} \cup
  \set{\edge{b}})$ is not an AF of $F_1$ and $F_2$.
  We prove that
  $F_2 \div (E \setminus \set{\edge{a}, \edge{c}, \edge{v}} \cup
  \set{\edge{b}, \edge{u}})$ is an AF of $F_1$ and $F_2$ and that
  $\size{E \cap \set{\edge{a}, \edge{c}, \edge{v}}} \ge\nobreak 2$ in this case.
  The latter implies that $\size{E \setminus \set{\edge{a}, \edge{c}, \edge{v}}
    \cup \set{\edge{b}, \edge{u}}} \le \size{E}$, that is,
  $F_2 \div (E \setminus \set{\edge{a}, \edge{c}, \edge{v}} \cup
  \set{\edge{b}, \edge{u}})$ is an MAF of $F_1$ and $F_2$.

  If $F_2 \div (E \setminus \set{\edge{a}, \edge{c}, \edge{v}} \cup
  \set{\edge{b}})$ is not an AF of $F_1$ and $F_2$, then either two of its
  components overlap in $F_1$ or it contains a triple incompatible with $F_1$.
  First consider the case of overlapping components.
  Observe that $F_2 \div (E \cup \set{\edge{b}})$ is an AF of $F_1$ and $F_2$
  because it is a refinement of $F_2 \div E$.
  The only component of $F_2 \div (E \setminus \set{\edge{a}, \edge{c},
    \edge{v}} \cup \set{\edge{b}})$ that is not a component of $F_2 \div
  (E \cup \set{\edge{b}})$ is the one containing $a$ and~$c$.
  Call this component~$C$.
  Thus, if two components of $F_2 \div (E \setminus \set{\edge{a}, \edge{c},
    \edge{v}} \cup \set{\edge{b}})$ overlap in $F_1$, one of them must
  be $C$.
  Call the other component $C'$.
  For any two leaves $x$ and $y$ in $C$ such that $x, y \notin F_1^{p_a}$,
  the path $P$ between $x$ and $y$ also exists in
  $F_2 \div (E \cup \set{\edge{b}})$ and, thus, cannot overlap $C'$.
  Thus, w.l.o.g.\ $x \in F_1^{p_a}$.
  Now, if the edge $e$ shared by $P$ and $C'$ belonged to $F_1^{p_a}$,
  $P$ and $C'$ would also overlap in $F_2$ because
  $F_1^{p_a}$ is the same as the subtree of $F_2 \div \set{\edge{b}}$ with
  root~$u$.
  This, however, is impossible because
  $F_2 \div (E \setminus \set{\edge{a}, \edge{c}, \edge{v}} \cup
  \set{\edge{b}})$ is a forest of $F_2$.
  Thus, the edge $e$ shared by $P$ and $C'$ cannot belong to $F_1^{p_a}$,
  and we have $y \notin F_1^{p_a}$.
  This implies that the path from $x'$ to $y$, for any leaf
  $x' \in F_2^a \cup F_2^c$, includes $e$.
  Therefore, since $F_2 \div (E \cup \set{\edge{b}})$ is an AF of $F_1$
  and~$F_2$, we have $x' \noreach[F_2 - (E \cup \set{\edge{b}})] y$, for every
  leaf $x' \in F_2^a \cup F_2^c$.
  Since $x \reach[F_2 \div (E \setminus \set{\edge{a}, \edge{c},
    \edge{v}} \cup \set{\edge{b}})] y$, the path from $u$ to $y$ in $F_2$
  contains no edge in $E$.
  Thus, since $x' \noreach[F_2 - (E \cup \set{\edge{b}})] y$, for all
  leaves $x' \in F_2^a \cup F_2^c$, the choice of $E$ and
  Lemma~\ref{lem:edge-shift} imply that $E$ must include $\edge{c}$ and at least
  one of $\edge{a}$ or $\edge{v}$, that is,
  $\size{E \cap \set{\edge{a}, \edge{c}, \edge{v}}} \ge 2$.

  In $F_2 - (E \setminus \set{\edge{a}, \edge{c},\edge{v}} \cup
  \set{\edge{b}, \edge{u}})$, $C$ is split into two components
  $C_1 = C \cap F_2^u$ and $C_2 = C \setminus F_2^u$.
  All other components are the same as in $F_2 - (E \setminus \set{\edge{a},
    \edge{c},\edge{v}} \cup \set{\edge{b}})$.
  Since $x, y \in F_1^{p_a}$, for all leaves $x, y \in C_1$, and $x, y \notin
  F_1^{p_a}$, for all leaves $x, y \in C_2$, the same argument as in the
  previous paragraph shows that neither $C_1$ nor $C_2$ overlaps a component
  $C' \notin \set{C_1, C_2}$.
  $C_1$ and $C_2$ do not overlap either because $C_1 \subseteq F_1^{p_a}$ and
  $C_2 \cap F_1^{p_a} = \emptyset$.
  Thus, no two components of $F_2 - (E \setminus \set{\edge{a}, \edge{c},
    \edge{v}} \cup \set{\edge{b}, \edge{u}})$ overlap in~$F_1$.

  Now assume $F_2 \div (E \setminus \set{\edge{a}, \edge{c},
    \edge{v}} \cup \set{\edge{b}})$ contains a triple incompatible with $F_1$.
  Then, once again, this triple has to be part of $C$ and must involve a leaf
  in $F_2^a \cup F_2^c$ and a leaf not in $F_2^a \cup F_2^c$ because any other
  triple is either a triple of $F_2 \div (E \cup \set{\edge{b}})$ or a triple of
  $F_1^{p_a}$; in either case, it is a triple of $F_1$.
  $F_2 \div (E \setminus \set{\edge{a}, \edge{c}, \edge{v}} \cup \set{\edge{b},
    \edge{u}})$ cannot contain a triple with one leaf in $F_2^a \cup F_2^c$
  and one leaf not in $F_2^a \cup F_2^c$ because the path between any two such
  leaves includes~$\edge{u}$.
  Thus, $F_2 \div (E \setminus \set{\edge{a}, \edge{c}, \edge{v}} \cup
  \set{\edge{b}, \edge{u}})$ contains no triples incompatible with $F_1$.
  Since we have just shown that no two components of
  $F_2 \div (E \setminus \set{\edge{a}, \edge{c},
    \edge{v}} \cup \set{\edge{b}, \edge{u}})$ overlap in $F_1$,
  $F_2 \div (E \setminus \set{\edge{a}, \edge{c},
    \edge{v}} \cup \set{\edge{b}, \edge{u}})$ is an AF of $F_1$ and $F_2$.

  It remains to prove that $\size{E \cap \set{\edge{a}, \edge{c}, \edge{v}}} \ge
  2$ if $F_2 \div (E \setminus \set{\edge{a}, \edge{c},
    \edge{v}} \cup \set{\edge{b}})$ contains a triple $xy|z$ incompatible with
  $F_1$.
  Since this triple needs to involve a leaf in $F_2^a \cup F_2^c$ and one not
  in $F_2^a \cup F_2^c$, we have (i) $x, y \in F_2^a \cup F_2^c$ and
  $z \notin F_2^a \cup F_2^c$, (ii)~$x \in F_2^a \cup F_2^c$ and
  $y, z \notin F_2^a \cup F_2^c$ or (iii) $x, y \notin F_2^a \cup F_2^c$
  and $z \in F_2^a \cup F_2^c$.
  The first case cannot arise because $x, y \in F_1^{p_a}$ and
  $z \notin F_1^{p_a}$ in this case, that is, $xy|z$ is also a triple of $F_1$.

  For the second case, assume for the sake of contradiction that
  $\size{E \cap \set{\edge{a}, \edge{c}, \edge{v}}} = 1$, and assume w.l.o.g.\
  that $x \in F_2^a$.
  Since every triple $x'y|z$ with $x' \in F_2^a$ would also be incompatible
  with $F_1$, $F_2 \div (E \setminus \set{\edge{b}})$ cannot contain such
  a triple.
  Hence, the choice of $E$ and Lemma~\ref{lem:edge-shift} imply that
  $E \cap \set{\edge{a}, \edge{v}} \ne \emptyset$ and, therefore,
  $\edge{c} \notin E$.
  As in the proof of Lemma~\ref{lem:spr:3-way}, this implies that there exists a
  leaf $c' \in F_2^c$ such that $c' \reach[F_2 - E] c \reach[F_2 - E] u$,
  by the choice of $E$ and Lemma~\ref{lem:edge-shift}.
  Since $xy|z$ is a triple of $F_2 \div (E \setminus \set{\edge{a}, \edge{c},
    \edge{v}} \cup \set{\edge{b}})$, we have
  $y \reach[F_2 - E] u \reach[F_2 - E] z$.
  Hence, $c'y|z$ is a triple of $F_2 \div E$ and this triple is incompatible
  with $F_1$ because $xy|z$ is, $x, c' \in F_1^{p_a}$, and
  $y, z \notin F_1^{p_a}$.
  This is a contradiction, that is,
  $\size{E \cap \set{\edge{a}, \edge{c}, \edge{v}}} \ge 2$.

  In the last case, if we assume w.l.o.g.\ that $z \in F_2^a$, an analogous
  argument as for the second case shows that, if
  $\size{E \cap \set{\edge{a}, \edge{c}, \edge{v}}} = 1$, then
  $F_2 \div E$ contains a triple $xy|c'$ with $c' \in F_2^c$ and which is
  incompatible with $F_1$, which is again a contradiction.\qquad
\end{proof}

\begin{lemma}[Case~\ref{case:spr:cab}---Multiple Pendant Nodes]
  \label{lem:spr:cab}
  If $(a,c)$ is a sibling pair of~$F_1$, $a \reach[F_2] c$,
  and the path from $a$ to $c$ in $F_2$ has $q \ge 2$ pendant nodes
  $b_1, b_2, \dots, b_q$, then there exists an edge
  set $E$ of size $\ecut{T_1, T_2, F_2}$ (resp.\ $\aecut{T_1, T_2,
    F_2}$) and such that $F_2 \div E$ is an AF (resp.\ AAF) of $T_1$ and
  $T_2$ and either $E \cap \set{\edge{a}, \edge{c}} \ne \emptyset$ or
  $\set{\edge{b_1}, \edge{b_2}, \dots, \edge{b_q}} \subseteq E$.
\end{lemma}

\begin{proof}
  We prove the lemma by induction on $q$.
  For $q = 1$, the claim holds by Lemma~\ref{lem:spr:3-way},
  so assume $q > 1$ and the claim holds for $q-1$.
  Assume further that $b_1$ is the sibling of $a$.
  By Lemma~\ref{lem:spr:3-way}, there exists a set $E''$ of
  size $\ecut{T_1, T_2, F_2}$ and such that $F_2 \div E''$ is an AF of $F_1$
  and $F_2$ and $E'' \cap \set{\edge{a}, \edge{b_1}, \edge{c}} \ne
  \emptyset$.
  If $E'' \cap \set{\edge{a}, \edge{c}} \ne \emptyset$, we are done.
  Otherwise $\edge{b_1} \in E''$ and
  $\ecut{T_1, T_2, F_2'} = \ecut{T_1, T_2, F_2} - 1$, where
  $F_2' := F_2 \div \set{\edge{b_1}}$.
  In $F_2'$, the path from $a$ to $c$ has $q-1$ pendant nodes,
  namely $b_2, b_3, \dots, b_q$.
  Thus, by the induction hypothesis, there exists an edge set $E'$ of size
  $\ecut{T_1, T_2, F_2'}$ and such that $F_2' \div E'$
  is an AF of $F_1$ and $F_2'$ and
  $E' \cap \set{\edge{a}, \edge{c}} \ne \emptyset$ or
  $\set{\edge{b_2}, \edge{b_3}, \dots, \edge{b_q}} \subseteq E'$.
  The set $E := E' \cup \set{\edge{b_1}}$ has size
  $\size{E'} + 1 = \ecut{T_1, T_2, F_2}$, $F_2 \div E = F_2' \div
  E'$ is an AF of $F_1$ and $F_2$, and
  either $E \cap \set{\edge{a}, \edge{c}} \ne \emptyset$ or
  $\set{\edge{b_1}, \edge{b_2}, \dots, \edge{b_q}} \subseteq E$.

  The second claim of the lemma follows using the same arguments,
  since Lemma~\ref{lem:spr:3-way} holds for both AF and AAF.\qquad
\end{proof}

\section{Computing the Hybridization Number}

\label{sec:refinement}

In this section, we present our algorithm for computing the hybridization
number of two $X$-trees.
As in \S\ref{sec:fpt}, we focus on deciding whether
$\dhyb{T_1, T_2} \le k$, as $\dhyb{T_1, T_2}$ can be computed by trying
increasing values of~$k$ and this does not increase the running time by
more than a constant factor.
Also as in \S\ref{sec:fpt}, it will be obvious from the
description of our algorithm that it produces a corresponding AAF when it
answers ``yes''.

Every AAF of $T_1$ and $T_2$ can be computed by first computing an AF $F$ of
$T_1$ and $T_2$ and then cutting additional edges in $F$ as necessary to break
cycles in $F$'s cycle graph $G_F$.
This suggests the following strategy to decide whether $\dhyb{T_1, T_2} \le k$:
We modify the MAF algorithm from \S\ref{sec:fpt} called with
parameter $k$.
Note that this algorithm may find AFs that are not maximum when $k >
\dspr{T_1, T_2}$, so we do not restrict our search to refinements of MAFs.
For every invocation $\alg{F_1, F_2, k''}$ of the
algorithm that would return ``yes''
in Step~\ref{case:success}, $F_2$ is an AF of $T_1$ and $T_2$ obtained by
cutting $k' := k - k''$ edges.
$F_2$ may not be an AAF of $T_1$ and $T_2$, but it may be possible to break
all cycles in $\cyc{F_2}$ by cutting at most $k''$ additional edges,
in which case $\dhyb{T_1, T_2} \le k' + k'' = k$.
Thus, instead of unconditionally returning ``yes'' in Step~\ref{case:success},
we invoke a second algorithm $\balg{F_2, k}$, which decides whether $F_2$
can be refined to an AAF of $T_1$ and $T_2$ with at most $k + 1$ components,
and return its answer.
We use $\aalg{F_1, F_2, k''}$ to denote an invocation of this modified
MAF algorithm.
We refer to the part of the algorithm consisting of these invocations
$\aalg{F_1, F_2, k''}$ as the \emph{branching phase} of the algorithm and
to the part that consists of the invocations $\balg{F_2, k}$ as the
\emph{refinement phase}.
We also refer to a single invocation $\balg{F_2, k}$ as a \emph{refinement
  step}.  Note that this is not a linear process---our
algorithm performs a refinement step for each agreement forest it finds
and thus cycles between the branching phase and refinement phase.

Now let us call an invocation $\aalg{F_1, F_2, k''}$ \emph{viable} if there
exists an MAAF $F$ of $T_1$ and $T_2$ that is a forest of $F_2$.
Below we show how to ensure that there exists a viable invocation
$\aalg{F_1, F_2, k''}$ such that $F_2$ is an (not necessarily maximum)
AF of $T_1$ and $T_2$ if $\dhyb{T_1, T_2} \le k$.
The invocation $\balg{F_2, k}$ made by $\aalg{F_1, F_2, k''}$ returns ``yes'',
so the whole algorithm returns ``yes'' in this case.
If on the other hand $\dhyb{T_1, T_2} > k$, the algorithm either fails
to find an AF of $T_1$ and $T_2$ with at most $k + 1$ components or
none of the AFs it finds can be refined to an AAF with at most $k + 1$
components.
Thus, it returns ``no'' in this case.
In either case, the algorithm produces the correct answer.

So assume $\dhyb{T_1, T_2} \le k$.
We prove that every viable invocation $\aalg{F_1, F_2,\break k''}$ such that
$F_2$ is not an AF of $T_1$ and $T_2$ has a viable child invocation.
This immediately implies that there exists a viable invocation
$\aalg{F_1, F_2, k''}$ such that $F_2$ is an AF of $T_1$ and $T_2$ because
the top-level invocation $\aalg{T_1, T_2, k}$ is trivially viable and
the number of invocations the algorithm makes is finite.
If $F_2$ is not an AF of $T_1$ and $T_2$ in a viable invocation
$\aalg{F_1, F_2, k''}$, this invocation applies one of
Cases~\ref{case:spr:sc}--\ref{case:spr:cab}.
If it applies Case~\ref{case:spr:sc} or~\ref{case:spr:cab},
Lemmas~\ref{lem:spr:sc} and~\ref{lem:spr:cab} show that one of its child
invocations is viable.
In Case~\ref{case:spr:cob}, on the other hand, the child invocation
$\aalg{F_1, F_2 \div \set{\edge{b}}, k''-1}$ is not guaranteed to be viable.
The next lemma shows that either $\aalg{F_1, F_2 \div \set{\edge{b}}, k''-1}$
or $\aalg{F_1, F_2 \div \set{\edge{c}}, k''-1}$ is a viable invocation in this
case.
Thus, we modify the algorithm to make two invocations
$\aalg{F_1, F_2 \div \set{\edge{b}}, k''-1}$
and $\aalg{F_1, F_2 \div \set{\edge{c}}, k''-1}$ in Case~\ref{case:spr:cob}.
Even with two recursive calls made in Case~\ref{case:spr:cob},
the recurrence bounding the number of recursive calls
made by the algorithm in the proof of Lemma~\ref{lem:invocations} remains
dominated by Case~\ref{case:spr:cab}.
Thus, the algorithm continues to make $\OhL{2.42^k}$ recursive calls.

\begin{lemma}[Case~\ref{case:spr:cob}---One Pendant Node---MAAF]
  \label{lem:hyb:cob}
  If $(a, c)$ is a sibling pair of~$F_1$, $a \reach[F_2] c$, and
  the path from $a$ to $c$ in $F_2$ has only one pendant node $b$, then there
  exists an edge set $E$ of size $\aecut{T_1, T_2, F_2}$ and such that $F_2 \div
  E$ is an AAF of $T_1$ and $T_2$ and
  $E \cap \set{\edge{b}, \edge{c}} \ne \emptyset$.
\end{lemma}

\begin{proof}
  Let $E'$ be an edge set of size $\aecut{T_1, T_2, F_2}$ and such
  that $F_2 \div E'$ is an AAF of $T_1$ and $T_2$.
  Assume further that there is no such set containing more edges from
  $\set{\edge{a}, \edge{b}, \edge{c}}$ than $E'$ and that $b$ is $a$'s sibling
  in $F_2$.
  By Lemma~\ref{lem:spr:3-way}, $E' \cap \set{\edge{a}, \edge{b}, \edge{c}} \ne
  \emptyset$.
  If $E' \cap \set{\edge{b}, \edge{c}} \ne \emptyset$, we are done.
  So assume $E' \cap \set{\edge{b}, \edge{c}} = \emptyset$ and, hence,
  $\edge{a} \in E'$.
  As in the proof of Lemma~\ref{lem:spr:cob}, let $v
  = \parent{a} = \parent{b}$ and $u = \parent{c} = \parent{v}$.
  If $\set{\edge{a}, \edge{v}} \subseteq E'$, Lemma~\ref{lem:edge-shift}
  implies that we can replace $\edge{v}$ with $\edge{b}$ in $E'$
  without changing $F_2 \div E'$.
  This contradicts the choice of $E'$, so $\edge{v} \notin E'$.
  As in the proof of Lemma~\ref{lem:spr:3-way}, the choice of $E'$ and
  Lemma~\ref{lem:edge-shift} imply that there exist leaves
  $b' \in F_2^b$ and $c' \in F_2^c$ such that $b' \reach[F_2 - E'] b$
  and $c' \reach[F_2 - E'] c$ because $E' \cap \set{\edge{b}, \edge{c}} =
  \emptyset$.
  Now let $E := E' \setminus \set{\edge{a}} \cup \set{\edge{b}}$.
  We have $\size{E} = \size{E'} = \aecut{T_1, T_2, F_2}$ and $\edge{b} \in E$.
  Moreover, since $E' \cap \set{\edge{a}, \edge{c}, \edge{v}} = \set{\edge{a}}$,
  the proof of Lemma~\ref{lem:spr:cob} shows that $F_2 \div E$ is an AF of
  $T_1$ and $T_2$.
  Next we show that $F_2 \div E$ is acyclic.

  Since $F_2 \div E$ and $F_2 \div E'$ are agreement forests of $T_1$
  and $T_2$, the mapping $\phi_1(\cdot)$ maps each node of these two
  forests to a corresponding node in $T_1$.
  However, a node $x \in F_2$ that belongs to both $F_2 \div E$ and
  $F_2 \div E'$ may map to different nodes in $T_1$ if it has different
  sets of labelled descendant leaves in $F_2 \div E$ and $F_2 \div E'$.
  For the remainder of this proof, we use $\phi_1(x)$ to denote
  the node in $T_1$ a node $x \in F_2$ maps to based on its labelled
  descendant leaves in $F_2 - E$, and $\phi_1'(x)$ to denote
  the node it maps to based on its labelled descendant leaves in $F_2 - E'$.

  Now assume for the sake of contradiction that $F_2 \div E$
  is not acyclic, and let $O$ be a cycle of $G_{F_2 \div E}$.
  We assume $O$ is as short as possible, which implies in particular that
  $O$ contains every component of $F_2 \div E$ at most once
  and that for any three consecutive components $C_i$, $C_{i+1}$, and
  $C_{i+2}$ in $O$ either $C_i$ is an ancestor of $C_{i+1}$ in $T_1$
  and $C_{i+1}$ is an ancestor of $C_{i+2}$ in $T_2$ or vice versa.
  Since $F_2 \div E'$ is acyclic, the root $r$ of at least one component
  in $O$ either is not a root in $F_2 \div E'$ or satisfies $\phi_1(r) \ne
  \phi_1'(r)$.
  The only root in $F_2 \div E$ that does not exist
  in $F_2 \div E'$ is a result of cutting edge $\edge{b}$ and is a
  descendant $z$ of $b$ in $F_2$.
  Let $C_z$ be the component of $F_2 \div E$ with root $z$.
  The only root in $F_2 \div E'$ that has a different set of labelled descendant
  leaves in $F_2 \div E$ is the root $u'$ of the component $C_u$ that contains
  $u$, and $\phi_1(u') \ne \phi_1'(u')$ only if $u' = u$.
  For any other component root~$x$, we have $\phi_1(x) = \phi_1'(x)$.
  Thus, any cycle $O$ in $G_{F_2 \div E}$ contains at least one of $C_u$ and
  $C_z$.
  Next we prove that no such cycle exists in $G_{F_2 \div E}$,
  by using the following five observations.
  \begin{enumerate}[label=(\roman{*}),leftmargin=0pt,itemindent=42pt]
  \item Since $u \reach[F_2 \div E'] z$ and $z$ is the only root of $F_2 \div E$
    that does not exist in $F_2 \div E'$, there is no root $x \notin \set{u,z}$
    of $F_2 \div E$ on the path from $u$ to $z$ in
    $T_2$.\label{item:T2-uz-no-root}
  \item Since $u \reach[F_2 \div E'] z$ and $z \in F_2^u$, we have
    $\phi_1'(z) \in T_1^{\phi_1'(u)}$.
    Any component $C_x$ with root $x$ such that $x \notin \set{u,z}$ satisfies
    $\phi_1'(x) = \phi_1(x)$.
    If $\phi_1'(x)$ belonged to the path from $\phi_1'(u)$ to $\phi_1'(z)$,
    then $C_x$ would overlap the component of $F_2 \div E'$ containing $u$
    in~$T_1$.
    Since $F_2 \div E'$ is a forest of $T_1$,
    no such component $C_x$ can exist.\label{item:T1-uz-no-root}
  \item Since $u \reach[F_2 \div E'] c'$, we have $c' \in T_1^{\phi_1'(u)}$ and,
    by the same arguments as in~\ref{item:T1-uz-no-root}, there is no root
    $x \notin \set{u,z}$ such that
    $\phi_1'(x) = \phi_1(x)$ belongs to the path from $c'$ to
    $\phi_1'(u)$ in $T_1$.\label{item:T1-uc-no-root}
  \item Since all labelled descendants of $u$ in $F_2 \div E$ belong to
    $F_2^a \cup F_2^c$, with at least one descendant in each of
    $F_2^a$ and $F_2^c$, we have $\phi_1(u) = \parent{a} = \parent{c}$.
    In particular, $c' \in T_1^{\phi_1(u)}$.
    Since $u$ has $c'$ and at least one labelled leaf in $F_2^b$
    as descendants in $F_2 \div E'$, $\phi_1'(u)$ is a proper ancestor
    of~$\phi_1(u)$.\label{item:u-uprime-ancestor}
  \item $\phi_1'(z) = \phi_1(z)$ is neither an ancestor nor a descendant
    of $\phi_1(u)$.
    The latter follows because $z$ has a labelled descendant
    leaf in $F_2 \div E$ that belongs to $F_2^b$, while all labelled
    descendant leaves of $\phi_1(u)$ belong to $F_2^a \cup F_2^c$.
    To see the former, observe that this would imply that $\phi_1'(z)$ is not a
    leaf and, hence, that there are two
    labelled descendant leaves $b_1$ and $b_2$ of $z$ in $F_2 \div E'$ such
    that $b_1, b_2 \in F_2^b$ and
    the path from $b_1$ to $b_2$ in $T_1$ includes $\phi_1'(z)$.
    Since $u \reach[F_2 \div E'] c'$ and $u \reach[F_2 \div E'] z$,
    this would imply that $F_2 \div E'$ contains the triple $b_1 b_2 | c'$,
    while these leaves would form the triple $b_1 c'|b_2$ or $b_2 c'|b_1$
    in~$T_1$.
    This is a contradiction because $F_2 \div E'$ is a forest
    of $T_1$.\label{item:u-z-no-ancestor}
  \end{enumerate}

  We now consider the different possible shapes of $O$.
  We use $C_{x_1}$ and $C_{x_2}$ to denote $C_u$'s predecessor and successor
  in $O$, respectively, and $C_{y_1}$ and $C_{y_2}$ to denote $C_z$'s
  predecessor and successor in $O$, respectively.
  First observe that $y_2 \ne u$ and, hence, $x_1 \ne z$.
  Indeed, $z \in F_2^u$, which implies that $y_2 = u$ only if
  $\phi_1(z)$ is an ancestor of~$\phi_1(u)$.
  By~\ref{item:u-z-no-ancestor}, this is impossible.

  If $y_1 = u$ (and $y_2 \ne u$), then $\phi_1(z) = \phi_1'(z)$ is an ancestor
  of $\phi_1(y_2) = \phi_1'(y_2)$ because, by~\ref{item:u-z-no-ancestor},
  $\phi_1(u)$ is not an ancestor of $\phi_1(z)$ and the edges in $O$ alternate
  between $T_1$ and $T_2$.
  By~\ref{item:T1-uz-no-root}, this implies that $\phi_1'(u)$ is an ancestor
  of $\phi_1'(y_2)$ in $T_1$.
  Also, for the predecessor $C_{x_1}$ of $C_u$ in $O$,
  $\phi_1'(x_1) = \phi_1(x_1)$ is an ancestor of $\phi_1(u)$ and, hence,
  by \ref{item:T1-uc-no-root} and~\ref{item:u-uprime-ancestor},
  an ancestor of $\phi_1'(u)$.
  This implies that we would obtain a cycle in
  $G_{F_2 \div E'}$ by removing $C_z$ from $O$, which contradicts
  that $F_2 \div E'$ is acyclic.
  This shows that $y_1 \ne u$.

  It remains to consider the case when $C_u$ and $C_z$ are not adjacent in $O$.
  In this case, all edges of $O$ except those incident to $C_u$ or $C_z$ exist
  also in $G_{F_2 \div E'}$ because $\phi_1'(x) = \phi_1(x)$, for every root
  $x \notin \set{u, z}$.
  Next we show that, if $C_u \in O$, then the edges $(C_{x_1}, C_u)$ and
  $(C_u, C_{x_2})$ also exist in $G_{F_2 \div E'}$, and if $C_z \in O$,
  then the edges $(C_{y_1}, C_u)$ and $(C_u, C_{y_2})$ exist in
  $G_{F_2 \div E'}$.
  Thus, by replacing $C_z$ with $C_u$ in $O$ (if $C_z \in O$), we obtain a cycle
  in $G_{F_2 \div E'}$, a contradiction because $F_2 \div E'$ is acyclic.

  If $C_u \in O$, then either $\phi_1(x_1) = \phi_1'(x_1)$ is an ancestor of
  $\phi_1(u)$ and $x_2$ is a descendant of $u$, or $x_1$ is an ancestor of $u$
  and $\phi_1(x_2) = \phi_1'(x_2)$ is a descendant of~$\phi_1(u)$.
  In the former case, \ref{item:T1-uc-no-root} and~\ref{item:u-uprime-ancestor}
  imply that $\phi_1'(x_1)$ is an ancestor of $\phi_1'(u)$.
  In the latter case, \ref{item:u-uprime-ancestor}
  implies that $\phi_1'(x_2)$ is also a descendant of $\phi_1'(u)$.
  In both cases, the edges $(C_{x_1}, C_u)$ and $(C_u, C_{x_2})$ exist
  in $G_{F_2 \div E'}$.

  If $C_z \in O$, then either $\phi_1(y_1) = \phi_1'(y_1)$ is an ancestor
  of $\phi_1(z) = \phi_1'(z)$ and $y_2$ is a descendant of $z$, or
  $y_1$ is an ancestor of $z$ and $\phi_1(y_2) = \phi_1'(y_2)$ is a descendant
  of $\phi_1(z) = \phi_1'(z)$.
  In the former case, \ref{item:T1-uz-no-root} implies that $\phi_1'(y_1)$ is
  an ancestor of $\phi_1'(u)$ and $y_2$ is a descendant of $u$.
  In the latter case, \ref{item:T2-uz-no-root}
  and~\ref{item:T1-uz-no-root} imply that $y_1$ is an ancestor of $u$
  and $\phi_1'(y_2)$ is a descendant of $\phi_1'(u)$.
  In both cases, the edges $(C_{y_1}, C_u)$ and $(C_u, C_{y_2})$ exist in
  $G_{F_2 \div E'}$.

  We have shown how to construct a corresponding cycle in $G_{F_2 \div E'}$
  for every cycle $O \in G_{F_2 \div E}$.
  Since $F_2 \div E'$ is acyclic, this shows that $F_2 \div E$ is acyclic.\qquad
\end{proof}

We have thus shown that the branching phase of our algorithm will find
at least one (not necessarily maximal) AF $F$ that can be refined to an MAAF.

In the remainder of this section, we develop an efficient implementation of
$\balg{F, k}$.
To do so, we need several new ideas.
Each of the following sections discusses one of them.
The tools introduced in \S\ref{sec:hyb_graph}--\S\ref{sec:naive_refinement}
suffice to obtain a fairly simple implementation of $\balg{F, k}$ that leads to
an MAAF algorithm with running time $\OhL{9.68^k n}$.
\S\ref{sec:marking} and~\S\ref{sec:combinations} then introduce
two refinements that improve the algorithm's running time first to
$\OhL{4.84^k n}$ and then to $\OhL{3.18^k n}$.

In \S\ref{sec:hyb_graph}, we introduce an expanded cycle graph $\ecyc{F}$.
In $\ecyc{F}$, every node of $\cyc{F}$ is replaced with the component of $F$
it represents.
This allows us to identify exactly which edges in a component $C$ need to be cut
if we want to break a cycle in $\cyc{F}$ by removing $C$ from this cycle.
Moreover, if $F$ has $k'+1$ components, $\ecyc{F}$ contains only $2k'$ of the
edges of $\cyc{F}$.
This ensures that $\ecyc{F}$ has size $\Oh{n}$, which is the key to keeping
the MAAF algorithm's dependence on $n$ linear.

In \S\ref{sec:cycles}, we identify components of $F$
that are \emph{essential} for the cycles in $\ecyc{F}$ in the sense that at
least one essential component of each cycle $O$ in $\ecyc{F}$ has to be
eliminated to break $O$ (as opposed to replacing it with a shorter cycle).
For every essential component $C$ in such a cycle $O$, we identify one node in
$C$, called an \emph{exit node}, and show that there exists a component $C$ in
$O$ such that cutting all edges on the path from $C$'s exit node to $C$'s root
reduces $\aecut{T_1, T_2, F}$ by the number of edges cut.
We call the process of cutting these edges \emph{fixing} the exit node.

In \S\ref{sec:naive_refinement}, we show how to mark a subset of at most
$2k$ nodes in $F$ such that, if $F$ can be refined to an AAF of $T_1$ and $T_2$
with at most $k+1$ components, then fixing an appropriate subset of these
marked nodes produces such an AAF.
We call these marked nodes \emph{potential exit nodes} because they include
the exit nodes of all essential components of all cycles in $\ecyc{F}$.
We obtain a first simple implementation of $\balg{F, k}$ by testing for each
subset of potential exit nodes whether fixing it produces an AAF with at most
$k+1$ components.
Since this test can be carried out in linear time for each subset and there are
$2^{2k} = 4^k$ subsets to test, the running time of this implementation of
$\balg{F, k}$ is $\OhL{4^k n}$.
Since we make at most one invocation $\balg{F_2, k}$ per invocation
$\aalg{F_1, F_2, k''}$ of the MAAF algorithm and the MAAF
algorithm makes $\OhL{2.42^k}$ invocations $\aalg{F_1, F_2, k''}$, the resulting
MAAF algorithm has running time
$\OhL{2.42^k\parensL{n + 4^k n}} = \OhL{9.68^k n}$.

The bound of $2k$ on the number of potential exit nodes is obtained quite
naturally:
We can obtain $F$ from both $T_1$ and $T_2$ by cutting the edges
connecting the roots of the components of $F$ to their parents in these
trees.
There are at most $k$ component roots of $F$ that are not roots in $T_2$.
Each such component has two corresponding parent edges, one in $T_1$ and one in
$T_2$.
The potential exit nodes are essentially the top endpoints of these at most
$2k$ parent edges, and the top endpoints of the two parent edges of each
component root form a pair of potential exit nodes.
In \S\ref{sec:marking}, we augment the search for agreement forests
to annotate the component roots of each found
agreement forest $F$ with information about how $F$ was obtained from $T_2$.
Using this information, we mark one potential exit node in each pair of
potential exit nodes and show that it suffices to test for each subset of
marked potential exit nodes whether fixing it produces an AAF with at most
$k + 1$ components.
Since at most $k$ potential exit nodes get marked,
this reduces the cost of $\balg{F, k}$ to $\OhL{2^k n}$ and, hence, the running
time of the MAAF algorithm to $\OhL{4.84^k n}$.

In \S\ref{sec:combinations}, we tighten the analysis of our algorithm.
So far, we allowed both phases of the algorithm to cut $k$ edges.
However, $k$~is the \emph{total} number of edges we are allowed to cut.
Thus, if the number $k'$ of edges we cut to obtain an AF is large, there
are only $k'' := k - k'$ edges left to cut in the refinement step, allowing
us to restrict our attention to small subsets of marked potential exit nodes
and thereby reducing the cost of the refinement step substantially.
If, on the other hand, $k'$ is small, then there are only few marked potential
exit nodes and even trying all possible subsets of these nodes is not too
costly.
By analyzing this trade-off between the number of edges cut in each phase
of the algorithm, we obtain the claimed running time of $\OhL{3.18^k n}$.

\subsection{An Expanded Cycle Graph}

\label{sec:hyb_graph}

The expanded cycle graph $\ecyc{F}$ of an agreement forest $F$ of two rooted
phylogenies $T_1$ and $T_2$ is a supergraph $\ecyc{F} \supset F$ with the
same vertex set as $F$; see Figure~\ref{fig:g_f_new}.
Let $E_1$ and $E_2$ be minimal subsets of edges of $T_1$ and $T_2$ such that
$F = T_1 \div E_1 = T_2 \div E_2$.
In addition to the edges of~$F$, $\ecyc{F}$~contains one \emph{hybrid edge} per
edge in $E_1 \cup E_2$.
To define these edges, we define mappings from nodes of $F$ to nodes of $T_1$
and $T_2$ and vice versa.
As in the definition of the original cycle graph $\cyc{F}$ in
\S\ref{sec:prelim}, we map each node $x$
in $F$ to nodes $\phi_1(x)$ in $T_1$ and $\phi_2(x)$ in $T_2$
such that $\phi_i(x)$ is the lowest common ancestor of all labelled leaves in
$T_i$ that are descendants of $x$ in $F$.
For the reverse direction, we define a function $\phi_i^{-1}(\cdot)$ mapping
nodes in $T_i$ to nodes in $F$;
$\phi_i^{-1}(x)$ is defined if and only if $x$ is labelled or
belongs to the path between two
labelled nodes $a$ and $b$ in $T_i$ such that $a \reach[F] b$.
In this case, $\phi_i^{-1}(x)$ is the node in $F$ that is the
lowest common ancestor of all labelled leaves $y$ in $T_i^x$ such that the path
between $x$ and $y$ does not contain any edges in $E_i$.
These mappings are well defined
in the sense that $\phi^{-1}_i(\phi_i(x)) = x$, for all $x \in F$ and $i
\in \set{1,2}$.

The hybrid edges in $\ecyc{F}$ are now defined as follows.
There are two such edges per root node $y$ of~$F$, except $\rho$, one
induced by $T_1$ and one induced by $T_2$.
Let $z_i$ be the lowest ancestor of $\phi_i(y)$ in $T_i$ such that
$\phi^{-1}_i(z_i)$ is defined.
Then $\parensL{\phi^{-1}_1(z_1),y}$ is a \emph{$T_1$-hybrid edge}
and $\parensL{\phi^{-1}_2(z_2),y}$ is a \emph{$T_2$-hybrid edge}.
See Figure~\ref{fig:g_f_new} for an illustration of these edges.
Note that neither $\phi_1^{-1}(z_1)$ nor $\phi_2^{-1}(z_2)$ is a root of $F$.
Our first lemma shows that the forest $F$ is an AAF of $T_1$ and $T_2$
if and only if $\ecyc{F}$ contains no cycles, that is, we can use
$\ecyc{F}$ in place of $\cyc{F}$ to test whether $F$ is acyclic.

\begin{figure}[t]
  \centering
  \subfigure[\unskip\label{fig:g_f_trees}]{\includegraphics{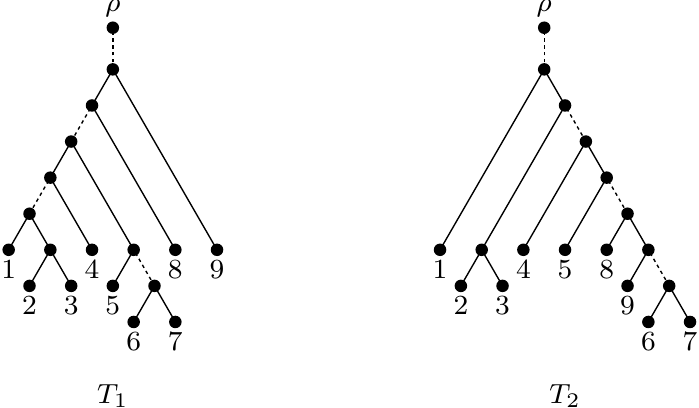}}
  \subfigure[\unskip\label{fig:g_f_old}]{\includegraphics{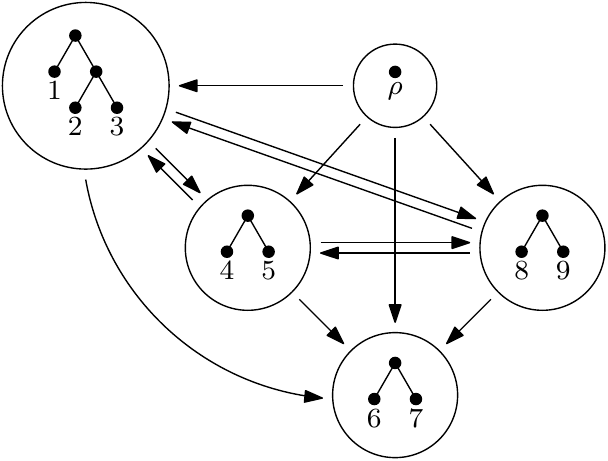}}%
  \hspace*{\stretch{1}}%
  \subfigure[\unskip\label{fig:g_f_new}]{\includegraphics{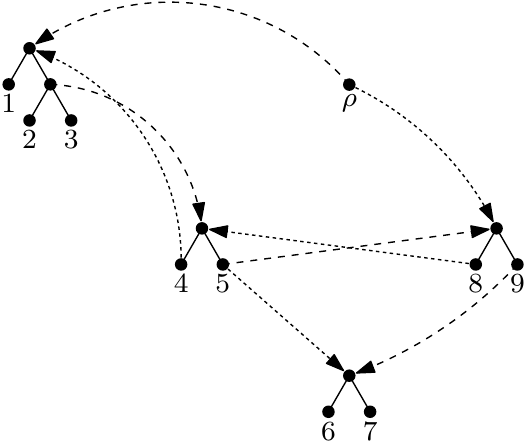}}%
  \caption{(a) Two trees $T_1$ and $T_2$.
    (b) An agreement forest $F$ of $T_1$ and $T_2$ obtained by cutting the
    dotted edges in $T_1$ and $T_2$, and its cycle graph $\cyc{F}$.
    The component of $F$ represented by each node of $\cyc{F}$ is drawn inside
    the node.
    (c) The expanded cycle graph $\ecyc{F}$.
    Dotted edges are $T_1$-hybrid edges, dashed ones are $T_2$-hybrid edges.}
  \label{fig:g_f}
\end{figure}

\begin{lemma}
  \label{lem:cyc=ecyc}
  $\ecyc{F}$ contains a cycle if and only if $\cyc{F}$ does.
\end{lemma}

\begin{proof}
  First observe that $\ecyc{F}$ can be obtained from $\cyc{F}$ by
  choosing a subset of the edges of $\cyc{F}$ and then replacing each
  vertex of $\cyc{F}$ with a component of $F$.
  Since the components of $F$ do not contain cycles, this shows that
  $\ecyc{F}$ is acyclic if $\cyc{F}$ is.

  Conversely, for two nodes $u$ and $v$ of $F$, $\ecyc{F}$ contains
  a path from $u$ to $v$ if $\phi_1(u)$ is an ancestor of $\phi_1(v)$
  or $\phi_2(u)$ is an ancestor of $\phi_2(v)$.
  Along with the fact that tree edges are directed away from the root of
  their component, this implies that every edge in $\cyc{F}$ can be
  replaced by a directed path in $\ecyc{F}$, so that $\ecyc{F}$ contains
  a cycle if $\cyc{F}$ does.\qquad
\end{proof}

In the remainder of this subsection, we show that $\ecyc{F}$ can be
constructed in linear time from $T_1$, $T_2$, and $F$, a fact we use
in our algorithms in \S\ref{sec:naive_refinement}, \S\ref{sec:marking},
and~\S\ref{sec:combinations}.

\begin{lemma}
  \label{lem:build-ecyc}
  The expanded cycle graph $\ecyc{F}$ of an agreement forest $F$
  of two rooted phylogenies $T_1$ and $T_2$ can be computed in linear time.
\end{lemma}

\begin{proof}
  Our construction of $\ecyc{F}$ starts with $F$ and then adds the hybrid edges.
  To add the hybrid edges induced by $T_1$, we perform a postorder traversal
  of $T_1$ that computes the mappings $\phi_1(\cdot)$ and $\phi_1^{-1}(\cdot)$,
  and the hybrid edges induced by $T_1$.
  A~similar postorder traversal of $T_2$ then computes $\phi_2(\cdot)$,
  $\phi_2^{-1}(\cdot)$, and the hybrid edges induced by $T_2$.

  We can assume each labelled node of $T_1$ or $T_2$
  stores a pointer to its counterpart in $F$ and vice versa.
  Thus, for each leaf $x$, $\phi_1(x)$, $\phi_2(x)$,
  $\phi^{-1}_1(x)$, and $\phi^{-1}_2(x)$ are given.
  In addition, we associate a list $L_x$ with each leaf $x$, where
  $L_x := \set{x}$ if $x$ is a root of~$F$, and $L_x = \emptyset$ otherwise.
  In general, after processing a node $x$, $L_x$ stores the set of roots
  of $F$ that map to descendants of $x$ and have proper ancestors of $x$
  as the tails of their $T_1$-hybrid edges.
  (It is not hard to see that this is the same ancestor of $x$, for every
  root in $L_x$.)

  After setting up this information for the leaves of $T_1$, the postorder
  traversal computes the same information for the nonleaf nodes of $T_1$
  and uses it to compute the $T_1$-hybrid edges in $\ecyc{F}$.
  For a nonleaf node $x$ with children $l$ and $r$, the mappings
  $\phi_1^{-1}(l)$ and $\phi_1^{-1}(r)$ and the root
  lists $L_l$ and $L_r$ of $l$ and $r$ are computed before processing $x$.
  Hence, we can use them to compute the mapping $\phi_1^{-1}(x)$ and the
  root list $L_x$.
  We distinguish four cases.

  If neither $\phi^{-1}_1(l)$ nor $\phi^{-1}_1(r)$ is undefined or a root of
  $F$, then they must have a common parent $p$ in $F$ (because $l$ and $r$
  are siblings in $T_1$ and $F$ is a forest of $T_1$).
  In this case, we set $\phi^{-1}_1(x) = p$ and $\phi_1(p) = x$.
  If $p$ is a root other than $\rho$, we set $L_x = \set{p}$; otherwise
  $L_x = \emptyset$.

  If both $\phi^{-1}_1(l)$ and $\phi^{-1}_1(r)$ are undefined or a root of $F$,
  then $\phi^{-1}(x)$ is undefined (as $x$ can belong to a path between two
  labelled nodes $a$ and $b$ such that $a \reach[F] b$ only if this is true
  for at least one of its children)
  and we set $L_x = L_l \cup L_r$.

  If only $\phi^{-1}_1(l)$ is undefined or a root of $F$, we set
  $\phi_1^{-1}(x) := \phi_1^{-1}(r)$
  and add a $T_1$-hybrid edge
  $\parensL{\phi_1^{-1}(x),y}$ to $\ecyc{F}$, for every root $y$ in $L_l$.
  Then we set $L_x = \emptyset$ ($x$~cannot be the image $\phi_1(x')$ of a root
  $x'$ of $F$ and $L_r = \emptyset$ in this case).

  The final case where only $\phi^{-1}_1(r)$ is undefined or a root of $F$ is
  symmetric to the previous case.

  It is easy to see that this procedure correctly constructs $\ecyc{F}$ because
  it directly follows the definition of $\ecyc{F}$.
  The running time of the algorithm is also easily seen to be linear.
  Indeed, computing the mappings $\phi_1^{-1}(x)$ and possibly $\phi_1(p)$ from
  $\phi_1^{-1}(l)$ and $\phi_1^{-1}(r)$ takes constant time per visited
  node~$x$, linear time in total.
  In the case when $L_x$ is computed as the union of $L_l$ and $L_r$, $L_l$
  and $L_r$ can be concatenated in constant time.
  In the case when we add a hybrid edge to~$\ecyc{F}$, for every node in $L_l$
  or $L_r$, this takes constant time per node, and we then pass an empty list
  $L_x$ to $x$'s parent.
  The latter implies that every
  root added to a list $L_x$ leads to the addition of exactly one hybrid edge
  to~$\ecyc{F}$.
  Since every node adds at most one root to $L_x$ that is not already present
  in $L_l$ or $L_r$, this shows that the addition of hybrid edges to $\ecyc{F}$
  also takes linear time in total for all nodes of $T_1$.
  The running time of the traversal of $T_2$ is bounded by $\Oh{n}$ using the
  same arguments.
  Hence, the entire algorithm takes linear time.\qquad
\end{proof}

One thing to note about the algorithm for constructing $\ecyc{F}$ is that it
does not require knowledge of the edge sets $E_1$ and $E_2$, even though we
used these sets to define $\ecyc{F}$.
This implies in particular that, even though there may be different edge sets
$E_1$ and $E_2$ such that $T_1 \div E_1 = T_2 \div E_2 = F$, all of them
lead to the same cycle graph---$\ecyc{F}$ is completely determined by $F$ alone.

\subsection{Essential Components and Exit Nodes}

\label{sec:cycles}

In this subsection, we define the essential components of a cycle in
$\ecyc{F}$ and their exit nodes.
Our goal is to prove that, if $F$ can be refined to an AAF of $T_1$
and $T_2$ with at most $k+1$ components, this is possible exclusively by cutting
the edges on the paths from exit nodes to the
roots of their components in~$F$.

Let $H_1$ be the set of $T_1$-hybrid edges in $\ecyc{F}$, and $H_2$
the set of $T_2$-hybrid edges in $\ecyc{F}$, and
assume $\ecyc{F}$ contains a cycle $O$.
Let $h_0, h_1, \ldots, h_{m-1}$ be the hybrid edges in $O$, and
consider the components $C_{0}, C_{1},
\ldots, C_{m-1}$ of $F$ connected by these edges.
More precisely, using index arithmetic modulo~$m$,
we assume the tail and head of edge $h_{i}$ belong to
components $C_{i}$ and $C_{i+1}$, respectively.
The cycle $O$ enters each component $C_{i}$ at its root and leaves it
at the tail of the edge $h_{i}$.
We say a component $C_i$ is \emph{essential for $O$} if
$h_{i-1} \in H_1$ and $h_{i} \in H_2$ or vice versa.
We say a component $C$ of $F$ is \emph{essential} if it is essential for
at least one cycle in $\ecyc{F}$.
A node $x$ of a component $C$ of $F$ is an \emph{exit node} of $C$ if
$C$ is an essential component $C_i$ for some cycle $O$ in $\ecyc{F}$ and
$x$ is the tail of edge $h_{i}$ in this cycle.
Figure~\ref{fig:cycle_g_f} illustrates these concepts.
Our first result in this subsection shows that there exists an exit node of an
essential component such that cutting its parent
edge in $F$ reduces $\aecut{T_1,T_2,F}$ by one, that is,
by cutting this edge, we make progress towards an MAAF of $T_1$ and~$T_2$.

\begin{figure}[t]
  \subfigure[\unskip\label{fig:cycle_trees}]
  {\includegraphics{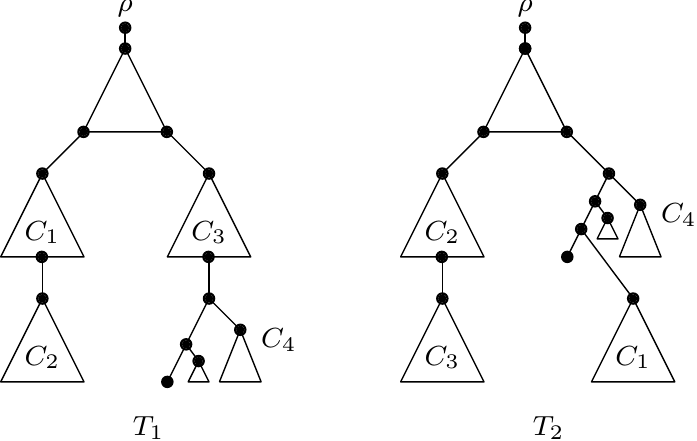}}%
  \hspace*{\stretch{1}}%
  \subfigure[\unskip\label{fig:cycle_forest}]
  {\includegraphics{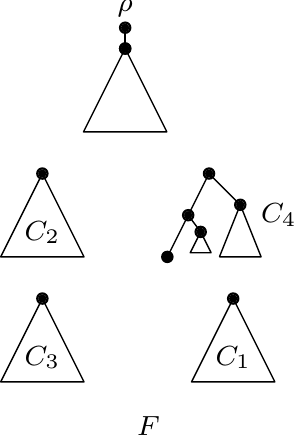}}%

  \subfigure[\unskip\label{fig:cycle_g_f}]
  {\includegraphics{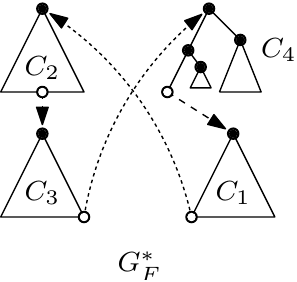}}%
  \hspace*{\stretch{1}}%
  \subfigure[\unskip\label{fig:cycle_fixed}]
  {\includegraphics{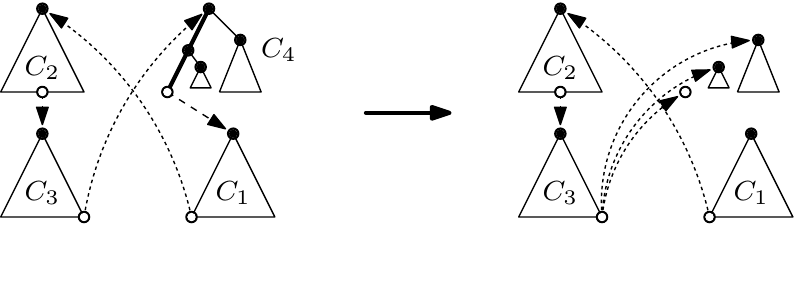}}%

  \caption{(a) Two trees $T_1$ and $T_2$.
    (b) An agreement forest $F$ of $T_1$ and $T_2$.
    (c) $\ecyc{F}$ (with $\rho$'s component removed for clarity)
    contains a cycle of length 4.
    White nodes indicate exit nodes.
    (d)~Fixing the exit node of component $C_4$ (cutting the bold edges)
    removes the cycle because
    none of the resulting subcomponents of $C_4$ is an ancestor of $C_1$
    in $T_2$.}
\end{figure}

\begin{lemma}
  \label{lem:fix_exit}
  Let $O$ be a cycle in $\ecyc{F}$, let $C_0, C_1, \dots, C_{m-1}$
  be its essential components, and let $v_i$ be the exit node of
  component $C_i$ in $O$, for all $0 \le i \le m-1$.
  Then $\aecut{T_1, T_2, F \div \set{\edge{v_i}}} = \aecut{T_1, T_2, F} - 1$,
  for some $0 \le i < m$.
\end{lemma}

\begin{proof}
  Let $E$ be an arbitrary edge set of size $\aecut{T_1, T_2, F}$ and such that
  $F' := F \div E$ is an AAF of $T_1$ and $T_2$.
  If $E \cap \set{\edge{v_0}, \edge{v_1}, \dots, \edge{v_{m-1}}} \ne \emptyset$,
  the lemma holds.
  If $E \cap \set{\edge{v_0}, \edge{v_1}, \dots, \edge{v_{m-1}}} = \emptyset$,
  we show that there exists an edge $f \in E$ such that
  $F' = F \div ( E \setminus \set{f} \cup \set{\edge{v_i}})$,
  for some $0 \le i < m$, which again proves the lemma.
  Let $r_i$ be the root of component $C_i$, for all $0 \le i < m$.
  To avoid excessive use of modulo notation in indices,
  we define $T_i$, $\phi_i(\cdot)$, etc.\ to be the same as
  $T_{2 - (i \bmod 2)}$, $\phi_{2 - (i \bmod 2)}(\cdot)$, etc.\ in the remainder
  of this proof.

  First suppose there exist leaves $a_i \in C_i^{v_i}$ and
  $c_i \in C_i \setminus C_i^{v_i}$ such that $a_i \reach[F'] c_i$, for
  all $0 \le i < m$, and let $l_i$ be the LCA of $a_i$ and $c_i$ in $F'$.
  Further, for every node $x \in F'$ and for $i \in \set{1,2}$,
  let $\phi_i(x)$ and $\phi_i'(x)$ be the nodes in $T_i$ $x$ maps to based
  on its descendants in $F$ and $F'$, respectively.
  Since $C_0, C_1, \dots, C_{m-1}$ are the essential components of $O$,
  $m$ is even and, w.l.o.g., the hybrid edge with head $r_i$ is
  $T_{i - 1}$-hybrid and the hybrid edge with tail $v_i$ is
  $T_i$-hybrid.
  This implies that the lowest ancestor $x_i$ of
  $\phi_i(r_{i+1})$ such that $\phi_i^{-1}(x_i)$
  is defined and belongs to $C_i$ satisfies
  $\phi_i^{-1}(x_i) = v_i$.

  Now observe that $\phi_i'(l_i)$ is a descendant of
  $\phi_i(r_i)$ and an ancestor of $x_i$ in $T_i$.
  The former follows because (i) the set of $l_i$'s descendants in $F'$ is a
  subset of $l_i$'s descendants in $F$ and, thus, $\phi_i'(l_i)$ is a descendant
  of $\phi_i(l_i)$, and (ii) $l_i$ is a descendant of $r_i$ in $F$ and, hence,
  $\phi_i(l_i)$ is a descendant of $\phi_i(r_i)$.
  The latter follows because
  $\phi_i'(a_i) = \phi_i(a_i)$ is a descendant
  of $x_i$,
  while $\phi_i'(c_i) = \phi_i(c_i)$ is not.
  Since $x_i$ is an ancestor of $\phi_i(r_{i+1})$, for all $i$,
  this implies that $\phi_i'(l_i)$ is an ancestor of
  $\phi_i'(l_{i+1})$, for all~$i$, which shows that the components of $F'$
  containing these nodes form a cycle in $\cyc{F'}$, contradicting that $F'$
  is acyclic.

  Thus, there exists a component $C_i$ such that $a \noreach[F'] c$,
  for all labelled leaves $a \in C_i^{v_i}$ and $c \in C_i \setminus C_i^{v_i}$.
  This in turn implies that either $a \noreach[F'] v_i$, for all labelled leaves
  $a \in C_i^{v_i}$, or $c \noreach[F'] v_i$, for all labelled leaves $c \in
  C_i \setminus C_i^{v_i}$.
  W.l.o.g., assume the former.
  We choose an arbitrary labelled
  leaf $a' \in C_i^{v_i}$ and let $f$ be the first edge in $E$ on the path from
  $v_i$ to $a'$.
  Since $a \noreach[F'] v_i$, for all $a \in C_i^{v_i}$, this edge $f$ and the
  edge $e = \edge{v_i}$ satisfy the conditions of Lemma~\ref{lem:edge-shift}
  and, hence, $F \div E = F \div (E \setminus \set{f} \cup \set{\edge{v_i}})$
  is an AAF of $T_1$ and $T_2$.\qquad
\end{proof}

The following corollary of Lemma~\ref{lem:fix_exit} shows that we can in fact
make progress towards an AAF by cutting \emph{all} edges on the path from an
appropriate exit node to the root of its component.
We call this \emph{fixing} the exit node.
Removing a cycle by fixing an exit node is illustrated
in Figure~\ref{fig:cycle_fixed}.

\begin{corollary}
  \label{cor:fix_exit}
  Let $O$ be a cycle in $\ecyc{F}$, let $C_0, C_1, \dots, C_{m-1}$ be its
  essential components, let $v_i$ be the exit node of component $C_i$ in $O$,
  let $F_i$ be the forest obtained from $F$ by fixing $v_i$, and let $\ell_i$
  be the length of the path in $C_i$ from $v_i$ to the root of $C_i$,
  for all $0 \le i \le m-1$.
  Then $\aecut{T_1, T_2, F_i} = \aecut{T_1, T_2, F} - \ell_i$, for
  some $0 \le i \le m - 1$.
\end{corollary}

\begin{proof}
  The proof is by induction on $\aecut{T_1, T_2, F}$.
  By Lemma~\ref{lem:fix_exit}, there exists some exit node $v_i$ such that
  $\aecut{T_1, T_2, F'} = \aecut{T_1, T_2, F} - 1$, where $F' := F \div
	\edge{v_i}$.
  Cutting $\edge{v_i}$ splits $C_i$ into two components $A_i$ and $B_i$
  containing the leaves in $C_i^{v_i}$ and in $C_i \setminus C_i^{v_i}$,
  respectively.

  If $\aecut{T_1, T_2, F} = 1$, then $\aecut{T_1, T_2, F'} = 0$.
  This implies that the path from $v_i$ to $r_i$ in $C_i$ cannot contain any
  edges apart from $\edge{v_i}$ because otherwise
  $C_0, C_1, \ldots, C_{i-1}, B_i,\break C_{i+1}, \ldots, C_{m-1}$
  would form a cycle in $\cyc{F'}$, that is,
  $\aecut{T_1, T_2, F'} > 0$.
  Thus, the corollary holds for $\aecut{T_1, T_2, F} = 1$.

  If $\aecut{T_1, T_2, F} > 1$, we can assume by induction that the corollary
  holds for $F'$.
  If \mbox{$\ell_i = 1$}, then the path from $v_i$ to $r_i$ consists of
  only $\edge{v_i}$, and the corollary holds for~$F$.
  Otherwise $C_0', C_1', \ldots, C_{i-1}', C_{i}', C_{i+1}', \ldots, C_{m-1}' :=
  C_0, C_1, \ldots, C_{i-1}, B_i, C_{i+1}, \ldots,\break C_{m-1}$ is a cycle
  $O'$ in $\cyc{F'}$.
  Note that, for $j \ne i$, the exit node $v_j'$ of $C_j'$ in $O'$ is~$v_j$;
  the exit node $v_i'$ of $C_i'$ is $v_i$'s sibling in~$C_i$.
  By the inductive hypothesis, there exists some $C'_j$, $0 \le j < m$, such
  that $\aecut{T_1, T_2, F_j'} = \aecut{T_1, T_2, F'} - \ell_j'$,
  where $F_j'$ is obtained from $F'$ by fixing $v_j'$
  and $\ell_j'$ is the length of the path from $v_j'$ to the root of~$C_j'$.
  In particular, $\ell_j' = \ell_j$, for $j \ne i$; and $\ell_i' = \ell_i - 1$.
  If $j \ne i$, we have $F_j' = F_j \div \set{\edge{v_i}}$ and
  $\aecut{T_1, T_2, F_j \div \set{\edge{v_i}}} = \aecut{T_1, T_2, F_j'} =
  \aecut{T_1, T_2, F'} - \ell_j' = \aecut{T_1, T_2, F} - \ell_j - 1$.
  Hence, $\aecut{T_1, T_2, F_j} \le \aecut{T_1, T_2, F} - \ell_j$.
  Since $F_j$ is obtained from $F$ by cutting $\ell_j$ edges, we also have
  $\aecut{T_1, T_2, F_j} \ge \aecut{T_1, T_2, F} - \ell_j$.
  Thus, the corollary holds in this case.
  If $j = i$, we have $F_j' = F_j$ and $\aecut{T_1, T_2, F_j} =
  \aecut{T_1, T_2, F_j'} =  \aecut{T_1, T_2, F'} - \ell_j' =
  \aecut{T_1, T_2, F} - \ell_j$.
  Thus, the corollary holds in this case as well.\qquad
\end{proof}

\subsection{Potential Exit Nodes and a Simple Refinement Algorithm}

\label{sec:naive_refinement}

In this subsection, we introduce the concept of potential exit nodes and
show that a first simple refinement algorithm can be obtained by testing for
each subset of potential exit nodes whether fixing these nodes produces
an AAF with at most $k+1$ components.

Given an agreement forest $F$ of $T_1$ and $T_2$, we mark all those nodes in
$F$ that are the tails of hybrid edges in $\ecyc{F}$.
Since this includes all exit nodes of $F$, we call these nodes
\emph{potential exit nodes}.
If $F$ has $k'$ components, there are $2(k'-1)$ potential exit nodes.
If $F$ is a forest produced by the branching phase of our algorithm, it has at
most $k+1$ components and, thus, at most $2k$ potential exit nodes.
The main result in this subsection is Lemma~\ref{lem:same_exit}, which
shows that the set of potential exit nodes of the forest obtained by fixing a
potential exit node in $F$ is a subset of $F$'s potential exit nodes.
We use this lemma to prove that, if $F$ can be refined to an AAF with at most
$k+1$ components, then fixing an appropriate subset of potential exit nodes
produces such a forest.

\begin{lemma}
  \label{lem:same_exit}
  Let $F$ be an agreement forest of two trees $T_1$ and $T_2$, let
  $V$ be the set of potential exit nodes of $F$, and let $v$ be an arbitrary
  node in $V$.
  Let $F'$ be the forest obtained from $F$ by fixing~$v$, and let $V'$ be the
  set of its potential exit nodes.
  Then $V' \subset V$.
\end{lemma}

\begin{proof}
  Since fixing $v$ removes $v$'s parent edge, $v$ is a root of $F'$, which
  implies that $v \notin V'$ because potential exit nodes are not component
  roots.
  Thus, $V' \ne V$, and it suffices to prove that $V' \subseteq V$.
  So let $u \in V'$, and let $(u,w)$ be a hybrid edge in $\ecyc{F'}$ with tail
  $u$.
  Assume w.l.o.g.\ that $(u, w)$ is a $T_1$-hybrid edge, let
  $\phi_1(\cdot)$ and $\phi_1^{-1}(\cdot)$ be defined as before with
  respect to $F$, and let $\phi_1'(\cdot)$ and $\phi_1'^{-1}(\cdot)$
  be the same mappings defined with respect to $F'$.
  By the definition of a hybrid edge, $w$ is the root of a
  component of $F'$ and $u = \phi_1'^{-1}(x)$, for the lowest proper ancestor
  $x$ of $\phi_1'(w)$ such that $\phi_1'^{-1}(x)$ is defined.

  Now let $E_1 \subset E_1'$ be edge sets such that $F = T_1 \div E_1$
  and $F' = T_1 \div E_1'$, and let $E$ be the set of edges cut
  in $F$ to fix $v$, that is, $F' = F \div E$.
  We prove that $a \reach[T_1 - E_1] x$ if and only if
  $a \reach[T_1 - E_1'] x$, for every labelled leaf $a \in T_1^x$.
  This implies in particular that $\phi_1'^{-1}(y) = \phi_1^{-1}(y)$,
  for all nodes $y \in T_1^x$ such that $x \reach[T_1 - E_1] y$.

  Clearly, if $a \reach[T_1 - E_1'] x$, then $a \reach[T_1 - E_1] x$
  because $E_1 \subset E_1'$.
  So assume $a \reach[T_1 - E_1] x$ but $a \noreach[T_1 - E_1']\nobreak x$,
  for some labelled leaf $a \in T_1^x$.
  Since $\phi_1'^{-1}(x)$ is defined, there exist labelled nodes $b$ and $c$
  such that $b \reach[F'] c$ and $x$ is on the path from $b$ to $c$ in
  $T_1 - E_1'$.
  This implies that $b \reach[T_1 - E_1'] x \reach[T_1 - E_1'] c$ and, hence,
  $b \reach[T_1 - E_1] x \reach[T_1 - E_1] c$.
  Together with $a \reach[T_1 - E_1] x$, this implies that
  $a$, $b$, and $c$ belong to the same connected component of $T_1 - E_1$ and,
  hence, to the same connected component of $F$, while $a$ belongs to a
  different connected component of $F'$ than $b$ and $c$.
  Now observe that, since $x$ is an ancestor of $a$ and
  is on the path from $b$ to $c$,
  the lowest common ancestor of $b$ and $c$ in $T_1$ is an ancestor of $a$.
  Since $F$ is a forest of $T_1$, this implies that the lowest common ancestor
  $l$ of $b$ and $c$ in $F$ also is an ancestor of $a$.
  Since $b \reach[F'] l \reach[F'] c$ and $a \noreach[F'] c$, the path from $a$
  to $l$ must contain at least one edge in $E$.
  By the choice of $E$, this implies that one of the child edges
  of $l$ also belongs to $E$ and, hence, that $b \noreach[F - E] c$, a
  contradiction because $F' = F \div E$ and $b \reach[F'] c$.

  To finish the proof, let $y$ be the first node after $x$ on the path from $x$
  to $\phi_1'(w)$ and such that $\phi_1^{-1}(y)$ is defined.
  Since $\phi_1'^{-1}(\phi_1'(w)) = w$, $\phi_1'^{-1}(\phi_1'(w))$ and, hence,
  $\phi_1^{-1}(\phi_1'(w))$ is defined, that is, such a node $y$ exists.
  If $x \noreach[T_1 - E_1] y$, then $\phi_1^{-1}(y)$ is a root of $F$ and
  $(\phi_1^{-1}(x), \phi_1^{-1}(y))$ is a hybrid edge in $\ecyc{F}$.
  Since $\phi_1^{-1}(x) = \phi_1'^{-1}(x) = u$, this proves that $u$ is also
  a potential exit node of $F$.
  If $x \reach[T_1 - E_1] y$, then $\phi_1'^{-1}(y) = \phi_1^{-1}(y)$,
  that is, $\phi_1'^{-1}(y)$ is defined.
  By the choice of $x$, this implies that $y = \phi_1'(w)$.
  Since $\phi_1'^{-1}(\phi_1'(w))$ is defined, there exists a leaf
  $a \in T_1^{\phi_1'(w)}$ such that $a \reach[T_1 - E_1'] \phi_1'(w)$
  and, hence, $a \reach[F'] w$ and $a \reach[T_1 - E_1] \phi_1'(w)$.
  Together with $\phi_1'(w) \reach[T_1 - E_1] x$, the latter implies that
  $a \reach[T_1 - E_1] x$, while $(u,w)$ being a hybrid edge implies
  that $u \noreach[F'] w$ and, hence, $a \noreach[F'] u$ and
  $a \noreach[T_1 - E_1'] x$.
  This is a contradiction, that is, the case $x \reach[T_1 - E_1] y$ cannot
  occur.\qquad
\end{proof}

By Corollary~\ref{cor:fix_exit}, if $F$ can be refined to an AAF $F'$ with at
most $k+1$ components, we can do so by fixing an appropriate exit node in
$F_0 = F$, then fixing an appropriate exit node in the resulting forest $F_1$,
and so on until we obtain $F'$.
Let $F = F_0, F_1, \dots, F_{k+1} = F'$ be the sequence of forests produced in
this fashion.
For $0 \le i \le k$, the exit nodes of $F_i$ are included in the set of $F_i$'s
potential exit nodes and, by Lemma~\ref{lem:same_exit}, these potential exit
nodes are included in the set of $F$'s potential exit node.
Thus, $F'$ can be obtained from $F$ by choosing an appropriate subset of
$F$'s potential exit nodes and fixing them.
Now it suffices to observe that fixing a subset of exit nodes one node at a time
produces the same forest as simultaneous cutting all edges in the union of the
paths from these exit nodes to the roots of their components in $F$.

This leads to the following simple refinement algorithm:
We mark the potential exit nodes in~$F$, which is easily done in linear time
as part of constructing $\ecyc{F}$.
Then we consider every subset of potential exit nodes.
For each such subset, we can in linear time identify the edges on the paths from
these potential exit nodes to the roots of their components, cut these edges and
suppress nodes with only one child, construct the expanded cycle graph
$\ecyc{F'}$ of the resulting forest $F'$, and test wether $F'$ has at most $k+1$
components and $\ecyc{F'}$ is acyclic.
We return ``yes'' as soon as we find a subset of potential exit nodes for which
this test suceeds.
If it fails for all subsets of potential exit nodes, we return ``no''.
If $F$ cannot be refined to an AAF with at most $k+1$ components, this test
fails for every subset of potential exit nodes.
Otherwise, as we have argued above, it will succeed for at least one subset of
potential exit nodes.
Thus, this implementation of $\balg{F, k}$~is correct.

If $F$ has $k' \le k+1$ components, there are at most $2^{2(k'-1)} \le 2^{2k} =
4^k$ subsets of potential exit nodes to test by $\balg{F, k}$.
Thus, the running time of $\balg{F, k}$ is $\OhL{4^k n}$.
As we argued at the beginning of this section, using this implementation of
$\balg{F, k}$ for the refinement phase of our MAAF algorithm results in a
running time of $\OhL{2.42^k\parens{n + 4^k n}} = \OhL{9.68^k n}$, and we
obtain the following result.

\begin{theorem}
\label{thm:naive}
For two rooted trees $T_1$ and $T_2$ and a parameter $k$, it takes
$\OhL{9.68^k n}$ time to decide whether $\aecut{T_1, T_2, T_2} \le k$.
\end{theorem}

\subsection{Halving the Number of Potential Exit Nodes}

\label{sec:marking}

In this subsection, we show how to mark half of the at most $2k$ potential exit
nodes defined in \S\ref{sec:naive_refinement} and show that it suffices to
test for every subset of \emph{marked} potential exit nodes whether fixing it
produces an AAF of $T_1$ and $T_2$ with at most $k+1$ components.
Since this reduces the number of subsets to be tested from $4^k$ to $2^k$,
the running time of the refinement step is reduced to $\OhL{2^k n}$, and
the running time of the entire MAAF algorithm is reduced to $\OhL{4.84^k n}$.

In general, the result of marking only a subset of potential exit nodes is
that we may obtain an AF $F$ of $T_1$ and $T_2$ that \emph{can}
be refined to an AAF of $T_1$ and $T_2$ with at most $k+1$ components
but \emph{cannot} be refined to such an AAF by fixing any subset of
the potential exit nodes marked in~$F$.
Intuitively, the reason why this is not a problem is that,
whenever we reach such an AF $F$ where a potential exit node $u$ should be fixed
but is not marked, there exists a branch in the branching phase's search for
AFs that cuts a subset of the edges cut to produce $F$ and then cuts $\edge{u}$.
Thus, if it is necessary to fix $u$ in $F$ to obtain an AAF $F'$ of $T_1$ and
$T_2$ with at most $k+1$ components, there exists an alternate route to
obtain the same AAF $F'$ by first producing a different AF $F''$
and then refining it.
While this is the intuition, it is in fact possible that our algorithm is not
able to produce $F'$ from $F''$ either.
What we do prove is that, if $\aecut{T_1, T_2, T_2} \le k$, then there exists a
``canonical'' AF $F_C$ produced by the branching phase of our algorithm and
which can be refined to an AAF $F_C'$ of $T_1$ and $T_2$ with at most $k+1$
components by fixing a subset of the marked potential exit nodes in $F_C$.

We accomplish the marking of potential exit nodes as follows.
The branching phase assigns a tag ``$T_1$'' or ``$T_2$'' to each component root
other than $\rho$ of each AF $F$ it produces.
After constructing $\ecyc{F}$, the refinement step marks a potential exit node
$u$ if there exists a $T_i$-hybrid edge $(u, w)$ in $\ecyc{F}$ such that $w$'s
tag is ``$T_i$''.
Finally, the refinement step checks whether an AAF of $T_1$ and $T_2$ with at
most $k+1$ components can be obtained from $F$ by fixing a subset of the marked
potential exit nodes.

To tag component roots during the branching phase of the algorithm, we augment
the three cases of Step~\ref{case:spr:non-sibling} to tag the bottom endpoints
of the edges they cut in $F_2$.
When a tagged node $x$ loses a child $l$ by cutting its parent
edge $\edge{l}$, $x$ is contracted into its other child $r$;
in this case, $r$~inherits $x$'s tag.
This ensures that at any time exactly the roots in the current forest $F_2$
are tagged.
The following is the pseudocode of the MAAF algorithm, which shows the tags
assigned to the component roots produced in Step~\ref{case:hyb:non-sibling}.
Note that Case~\ref{case:hyb:cob} has an additional branch that cuts both
$\edge{a}$ and $\edge{c}$.
This is necessary to ensure we find an AAF of $T_1$
and $T_2$ with at most $k+1$ components in spite of considering only subsets
of marked potential exit nodes in the refinement phase (if such an AAF exists).
In the description of the algorithm, we use $k$ to denote the parameter passed
to the current invocation (as in the MAF algorithm), and $k_0$ to denote the
parameter of the top-level invocation $\aalg{T_1, T_2, k_0}$.
Thus, $k_0+1$ is the number of connected components we allow the final AAF
to have.

\begin{enumerate}[leftmargin=*,widest=6]
\item \label{case:hyb:abort} (Failure) If $k < 0$, there is no subset $E$
  of at most $k$ edges of $F_2$ such that $F_2 - E$ yields an AF of
  $T_1$ and $T_2$: $\aecut{T_1, T_2, F_2} \ge 0 > k$.
  Return ``no'' in this case.
\item \label{case:hyb:success} (Refinement) If $\size{R_t}
  \le 2$, then $F_2 = \dot{F}_2 \cup F$ is an AF of
  $T_1$ and $T_2$.
  Invoke an algorithm $\balg{F_2,k_0}$ that decides whether $F_2$ can be refined
  to an AAF of $T_1$ and $T_2$ with at most $k_0+1$ components.
  Return the answer returned by $\balg{F_2, k_0}$.
\item \label{case:hyb:singleton} (Prune maximal agreeing subtrees)
  If there is a node $r \in R_t$ that is a root in $\dot{F}_2$, remove $r$
  from $R_t$ and add it to $R_d$, thereby moving the corresponding subtree
  of $\dot{F}_2$ to $F$; cut the edge $\edge{r}$ in $\dot{T}_1$
  and suppress $r$'s parent from $\dot{T}_1$; return to
  Step~\ref{case:hyb:success}.
  This does not alter $F_2$ and, thus, neither $\aecut{T_1,T_2,F_2}$.
  If no such root $r$ exists, proceed to Step~\ref{item:hyb:choose-sib-pair}.
\item \label{item:hyb:choose-sib-pair}Choose a sibling pair $(a,c)$ in
  $\dot{T}_1$ such that $a, c \in R_t$.
\item \label{case:hyb:sibling} (Grow agreeing subtrees) If $(a,c)$ is
  a sibling pair of $\dot{F}_2$, remove $a$ and $c$ from $R_t$; label
  their parent in both trees with $(a,c)$ and add it to $R_t$; return to
  Step~\ref{case:hyb:success}.
  If $(a,c)$ is not a sibling pair of $\dot{F}_2$, proceed to
  Step~\ref{case:hyb:non-sibling}.
\item \label{case:hyb:non-sibling} (Cut edges) Distinguish three
  cases:
  \begin{enumerate}[label=\theenumi.\arabic{*}.,leftmargin=*,ref=\theenumi.\arabic{*},widest=3]
  \item \label{case:hyb:sc} If $a \noreach[F_2] c$, make two recursive
    calls:
    \begin{enumerate}[label=(\roman{*}),widest=iii,leftmargin=*]
    \item $\aalg{F_1, F_2 \div \set{\edge{a}}, k-1}$ with $a$
      tagged with ``$T_2$'' in $F_2 \div \set{\edge{a}}$, and
    \item $\aalg{F_1, F_2 \div \set{\edge{c}}, k-1}$ with $c$
      tagged with ``$T_2$'' in $F_2 \div \set{\edge{c}}$.
    \end{enumerate}
  \item \label{case:hyb:cob} If $a \reach[F_2] c$ and the path from
    $a$ to $c$ in $\dot{F}_2$ has one pendant node $b$, swap the names of
    $a$ and $c$ if necessary to ensure that $b$ is $a$'s sibling.
    Then make three recursive calls (see Figure~\ref{fig:hyb:cases}):
    \begin{enumerate}[label=(\roman{*}),widest=iii,leftmargin=*]
    \item $\aalg{F_1, F_2 \div \set{\edge{b}}, k-1}$
      with $b$ tagged with ``$T_1$'' in $F_2 \div \set{\edge{b}}$,
    \item $\aalg{F_1, F_2 \div \set{\edge{c}}, k-1}$
      with $c$ tagged with ``$T_2$'' in $F_2 \div \set{\edge{c}}$, and
    \item $\aalg{F_1, F_2 \div \set{\edge{a},\edge{c}}, k-2}$
      with $c$ tagged with ``$T_1$'' and $a$ tagged with ``$T_2$'' in
      $F_2 \div \set{\edge{a}, \edge{c}}$.
    \end{enumerate}
  \item \label{case:hyb:cab} If $a \reach[F_2] c$ and the path from
    $a$ to $c$ in $\dot{F}_2$ has $q \ge 2$ pendant nodes $b_1, b_2,
    \ldots, b_q$, make three recursive calls:
    \begin{enumerate}[label=(\roman{*}),widest=iii,leftmargin=*]
    \item $\aalg{F_1, F_2 \div
        \set{\edge{b_1}, \edge{b_2}, \ldots, \edge{b_q}}, k-q}$
      with each node $b_i$, $1 \le i \le q$, tagged with ``$T_1$'' in
      $F_2 \div \set{\edge{b_1}, \edge{b_2}, \dots, \edge{b_q}}$,
    \item $\aalg{F_1, F_2 \div \set{\edge{a}}, k-1}$
      with $a$ tagged with ``$T_2$'' in $F_2 \div \set{\edge{a}}$, and
    \item $\aalg{F_1, F_2 \div \set{\edge{c}}, k-1}$
      with $c$ tagged with ``$T_2$'' in $F_2 \div \set{\edge{c}}$.
    \end{enumerate}
  \end{enumerate}
  Return ``yes'' if one of the recursive calls does; otherwise return
  ``no''.
\end{enumerate}

\begin{figure}[t]
  \centering
  \includegraphics{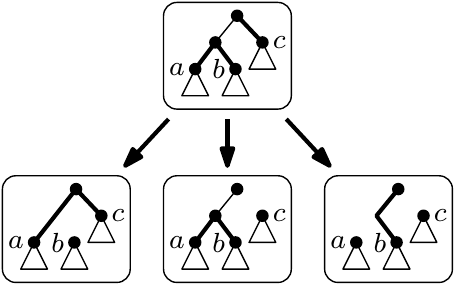}%
  \caption{Case~\ref{case:hyb:cob} of Step~\ref{case:hyb:non-sibling} of the
    rooted MAAF algorithm.
    (Cases~\ref{case:hyb:sc} and~\ref{case:hyb:cab} are as in the rooted MAF
    algorithm and are illustrated in Figure~\ref{fig:spr:cases}).
    Only $\dot{F}_2$ is shown.
    Each box represents a recursive call.}
  \label{fig:hyb:cases}
\end{figure}

To give some intuition behind the choice of tags in
Step~\ref{case:spr:non-sibling}, and as a basis for the correctness proof of
the algorithm, we consider a slightly modified algorithm that produces the same
set of forests:
When cutting an edge $\edge{x}$, $x \in \set{a, c}$,
in Step~\ref{case:hyb:non-sibling}, $x$ becomes the root of a component of $F_2$
that agrees with a subtree of $F_1$.
Hence, the first thing Step~\ref{case:hyb:singleton} of the next recursive call
does is to cut the parent edge of $x$ in $F_1$.
In the modified algorithm, we cut the parent edges of $x$ in \emph{both} $F_1$
and $F_2$ in Step~\ref{case:hyb:non-sibling} instead of postponing the cutting
of $x$'s parent edge in $F_1$ to Step~\ref{case:hyb:singleton} of the next
recursive call.

Now consider a labelled node $x$ of $F_2$, and let $y_1$ and $y_2$ be $x$'s
siblings in $F_1$ and~$F_2$, respectively.
If a case of Step~\ref{case:hyb:non-sibling} cuts the edge $\edge{x}$, $x$
becomes a root and, in the absence of further changes that eliminate $x$, $y_1$
or $y_2$ from the forest, $x$ is the head of a $T_1$-hybrid edge $(y_1,x)$ and
of a $T_2$-hybrid edge $(y_2,x)$, making $y_1$ and $y_2$ potential exit nodes
that may need to be fixed to obtain a certain AAF of $T_1$ and $T_2$.
The first step in fixing a potential exit node is to cut its parent edge, and an
alternate sequence of edge cuts that produce the same AAF starts by cutting this
parent edge instead of~$\edge{x}$.
Thus, if apart from cutting $\edge{x}$, the current case includes a branch that
cuts the parent edge of $y_1$ or $y_2$, we do not have to worry about fixing
this exit node in the branch that cuts~$\edge{x}$---there exits another branch
we explore that has the potential of leading to the same AAF.
To illustrate this idea, consider Case~\ref{case:hyb:sc}.
Here, when we cut~$\edge{a}$, $c$~becomes the tail of the $T_1$-hybrid edge
$(c,a)$ because $a$ and $c$ are siblings in~$F_1$.
Since the other branch of this case cuts $\edge{c}$, we do not have to
worry about fixing $c$ in the branch that cuts $\edge{a}$.
On the other hand, neither of the two cases considers cutting the parent edge
of $a$'s sibling $b$ in $F_2$, which is the tail of the $T_2$-hybrid edge with
head~$a$.
Thus, we need to give the refinement step the opportunity to fix $b$.
We do this by tagging $a$ with ``$T_2$'', which causes the refinement step to
mark $b$.
The same reasoning justifies the tagging of $c$ with ``$T_2$'' in the other
branch of this case and the tagging of $a$ and $c$ with ``$T_2$'' in
Case~\ref{case:hyb:cab}.
The tagging of every node $b_i$ with ``$T_1$'' in Case~\ref{case:hyb:cab}
is equally easy to justify: Cutting an edge $\edge{b_i}$ makes $b_i$ a root
in $F_2$.
Thus, either $b_i$ is itself a root of the final AF we obtain or it
is contracted into such a root $z$ after cutting additional edges; this root $z$
inherits $b_i$'s tag.
At the time we cut edge~$\edge{b_i}$, we do not know which descendant of $b_i$
will become this root $z$, nor whether any branch of our algorithm considers
cutting the parent edge of the tail of $z$'s $T_1$-hybrid edge.
On the other hand, the tail of $z$'s $T_2$-hybrid edge is either $a$ or~$c$,
and we cut their parent edges in the other two branches of
Case~\ref{case:hyb:cab}.
Unfortunately, the tagging in Case~\ref{case:hyb:cob} does not follow the same
intuition and is in fact difficult to justify intuitively.
The proof of Theorem~\ref{thm:maaf:fpt} below shows that the chosen tagging
rules lead to a correct algorithm.

We assume from here on that $\dhyb{T_1, T_2} \le k_0$
because otherwise $\balg{F, k_0}$ returns ``no'' for any AF $F$ the branching
phase may find, that is, the algorithm gives the correct answer when
$\dhyb{T_1, T_2} > k_0$.
Among the AFs of $T_1$ and $T_2$ produced by the branching phase, there may be
several that can be refined to an AAF of $T_1$ and $T_2$
with at most $k+1$ components.
We choose a \emph{canonical AF} $F_C$ from among these AFs.
The proof of
Theorem~\ref{thm:maaf:fpt} below shows that the potential exit nodes in $F_C$
that need to be fixed to obtain such an AAF are marked.
Since $F_C$ is produced by a sequence of recursive calls of procedure
$\aalg{\cdot,\cdot,\cdot}$, we can define $F_C$ by specifying the path to
take from the top-level invocation $\aalg{T_1, T_2, k_0}$ to the invocation
$\aalg{F_1, F_2, k}$ with $F_2 = F_C$.
We use $F_1^i$ and $F_2^i$ to denote the inputs to the $i$th invocation
$\aalgL{F_1^i,F_2^i,k_i}$ along this path.
We also compute an arbitrary numbering of the nodes of $T_1$ and denote
the number of $x \in T_1$ by $\nu(x)$.
This number is used as a tie breaker when choosing
the next invocation along the path of invocations that produce $F_C$.
The first invocation is of course $\aalg{T_1,T_2,k_0}$, that is,
$F_1^0 = T_1$ and $F_2^0 = T_2$.
So assume we have constructed the path up to the $i$th invocation
with inputs $F_1^i$ and $F_2^i$.
The $(i+1)$st invocation is made in Step~\ref{case:hyb:non-sibling} of the
$i$th invocation.
We say an invocation $\aalg{F_1, F_2, k}$ is a \emph{leaf invocation} if
$F_2$ is an AF of $T_1$ and~$T_2$.
Recall the definition of a viable invocation from the beginning of this section
and recall that $\aalg{T_1, T_2, k_0}$ is viable and that every viable
invocation that is not a leaf invocation has a viable child invocation.
If there is only one viable invocation made in
Step~\ref{case:hyb:non-sibling} of the $i$th invocation
$\aalgL{F_1^i, F_2^i, k_i}$, then we choose this invocation as the $(i+1)$st
invocation $\aalgL{F_1^{i+1}, F_2^{i+1}, k_{i+1}}$.
Otherwise we apply the following rules to choose
$\aalgL{F_1^{i+1}, F_2^{i+1}, k_{i+1}}$
from among the viable invocations made in Step~\ref{case:hyb:non-sibling}
of invocation $\aalgL{F_1^i, F_2^i, k_i}$.
We distinguish the three cases of Step~\ref{case:hyb:non-sibling}.

\textbf{Case~\ref{case:hyb:sc}.}
  In this case,
  $\aalgL{F_1^i \div \set{\edge{a}}, F_2^i \div \set{\edge{a}}, k_i - 1}$ and
  $\aalgL{F_1^i \div \set{\edge{c}},\break F_2^i \div \set{\edge{c}}, k_i - 1}$
  are both viable invocations.
  For $x \in \set{a, c}$, let $F_x$ be the agreement forest found by tracing
  a path from
  $\aalgL{F_1^i \div \set{\edge{x}}, F_2^i \div \set{\edge{x}}, k_i -1}$
  to a viable leaf invocation using recursive application of these rules, and
  let $E_x$ be an edge set such that $F_x = T_1 \div E_x$.
  Let $y$ be the sibling of $x$ in $F_1^i$ (i.e., $y = c$ if $x = a$ and vice
  versa).
  Now let $\phi_1(y)$ once again be the LCA in $T_1$ of all labelled
  leaves that are descendants of $y$ in $F_2^i$, and let $\phi_x(y)$ be the
  LCA in $F_x$ of all labelled leaves $l$ that are descendants of $\phi_1(y)$
  in $T_1$ and such that the path from $l$ to $\phi_1(y)$ in $T_1$ does not
  contain an edge in~$E_x$.
  In other words, $\phi_x(y)$ is the node of $F_x$ that $y$ is merged into
  by suppressing nodes during the sequence of recursive calls that produces
  $F_x$ from $F_2^i$.
  Finally, if $\phi_x(y)$ is the root of a component of $F_x$, let
  $\lambda_1(y) := \phi_1(y)$; otherwise let $\lambda_1(y)$ be the LCA in $T_1$
  of all labelled leaves that are descendants of the parent of $\phi_x(y)$
  in~$F_x$.
  In other words, if $\phi_x(y)$ is not a root in $F_x$, then $\lambda_1(y)$ is
  the node in $T_1$ where $\phi_x(y)$ and its sibling in $F_x$ are joined by an
  application of Step~\ref{case:hyb:sibling} in some recursive call on the path
  to~$F_x$.
  Now let $d_1(y) > 0$ be the distance from the root $\rho$ of $T_1$ to
  $\lambda_1(y)$ if $\lambda_1(y) \ne \phi_1(y)$, and $d_1(y) = 0$ otherwise.
  If $d_1(a) > d_1(c)$ or $d_1(a) = d_1(c)$ and $\nu(a) < \nu(c)$, we choose the
  invocation
  $\aalgL{F_1^i \div \set{\edge{a}}, F_2^i \div \set{\edge{a}}, k_i - 1}$
  as the $(i+1)$st invocation, that is, $F_C = F_a$.
  If $d_1(a) < d_1(c)$ or $d_1(a) = d_1(c)$ and $\nu(a) > \nu(c)$, we choose the
  invocation
  $\aalgL{F_1^i \div \set{\edge{c}}, F_2^i \div \set{\edge{c}}, k_i - 1}$ as the
  $(i+1)$st invocation, that is, $F_C = F_c$.
  This is illustrated in Figure~\ref{fig:f_c_sc}.

\textbf{Case~\ref{case:hyb:cob}.}
  In this case, if $\aalgL{F_1^i \div \set{\edge{a}, \edge{c}}, F_2^i \div
    \set{\edge{a}, \edge{c}}, k_i - 2}$ is viable, we choose it as the $(i+1)$st
  invocation.
  If the invocation $\aalgL{F_1^i \div \set{\edge{a}, \edge{c}},\break
  F_2^i \div \set{\edge{a}, \edge{c}}, k_i - 2}$ is not viable, then the
  invocations $\aalgL{F_1^i, F_2^i \div \set{\edge{b}}, k_i - 1}$ and
  $\aalgL{F_1^i \div \set{\edge{c}}, F_2^i \div \set{\edge{c}}, k_i - 1}$
  are both viable.
  In this case, we choose the latter as the $(i+1)$st invocation.

\textbf{Case~\ref{case:hyb:cab}.}
  Since there is more than one viable invocation in this case, at least one
  of the invocations
  $\aalgL{F_1^i \div \set{\edge{a}}, F_2^i \div \set{\edge{a}}, k_i - 1}$ and
  $\aalgL{F_1^i \div \set{\edge{c}}, F_2^i \div \set{\edge{c}},\break k_i - 1}$
  is viable.
  If exactly one of them is viable, we choose it to be the $(i+1)$st
  invocation.
  If both are viable, we define $\lambda_1(a)$ and $\lambda_1(c)$ as in
  Case~\ref{case:hyb:sc}.
  If $\lambda_1(a) \ne \lambda_1(c)$, we choose the $(i+1)$st invocation
  as in Case~\ref{case:hyb:sc}.
  If $\lambda_1(a) = \lambda_1(c)$,
  we define $\lambda_2(x)$ and $d_2(x)$, for $x \in \set{a, c}$, analogously to
  $\lambda_1(x)$ and $d_1(x)$ but using $\phi_2(\cdot)$ and $T_2$
  in place of $\phi_1(\cdot)$ and $T_1$.
  Now we choose the $(i+1)$st invocation as in Case~\ref{case:hyb:sc} but
  using $d_2(\cdot)$ instead of $d_1(\cdot)$.

\begin{figure}[t]
  \hspace*{\stretch{1}}%
  \subfigure[\unskip]{\includegraphics{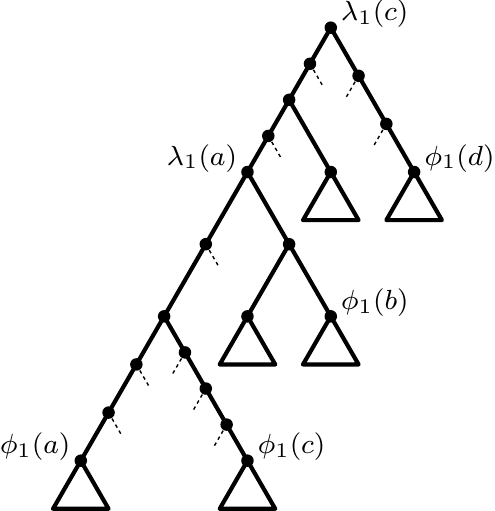}}%
  \hspace*{\stretch{1}}%
  \subfigure[\unskip]{\includegraphics{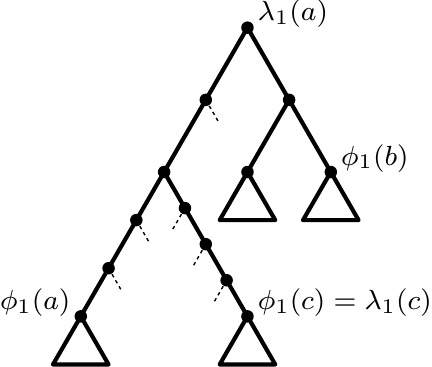}}%
  \hspace*{\stretch{1}}%
  \caption{Two applications of Case~\ref{case:hyb:sc} where we choose the
    invocation
    $\aalgL{F_1^i \div \set{\edge{a}}, F_2^i \div \set{\edge{a}},\break k_i - 1}$
    on the path to $F_C$.
    Both figures show the relevant portion of $T_1$.
    Dotted edges have been removed to obtain $F_1^i$, making $a$ and $c$
    siblings in $F_1^i$.
    The bold portion of $T_1$ yields $F_1^i$.
    Node $b$ is $a$'s sibling in $F_c$.
    In Figure~(a), $d$ is $c$'s sibling in $F_a$, and the highest node
    $\lambda_1(c)$ on the path from $\phi_1(c)$ to $\phi_1(d)$ is an
    ancestor of the highest node $\lambda_1(a)$ on the path from
    $\phi_1(a)$ to $\phi_1(b)$.
    Hence, $d_1(a) > d_1(c)$.
    In Figure (b), $\phi_a(c)$ is assumed to be a root of $F_a$.
    Hence $\phi_1(c) = \lambda_1(c)$ and $d_1(a) > 0 = d_1(c)$.}
  \label{fig:f_c_sc}
\end{figure}

\begin{lemma}
  \label{lem:f_c_refinement}
  If $\dhyb{T_1, T_2} \le k_0$, then $F_C$ can be refined to an AAF of $T_1$ and
  $T_2$ with at most $k_0+1$ components by fixing a subset of the marked
  potential exit nodes in $F_C$.
\end{lemma}

\begin{proof}
  Let $E$ be an edge set such that $F' := F_C \div E$ is an AAF of $T_1$ and
  $T_2$ with at most $k_0+1$ components.
  By Corollary~\ref{cor:fix_exit}, we can assume $E$ is the union of paths
  from a subset of potential exit nodes to the roots of their respective
  components in~$F_C$.
  These potential exit nodes may or may not be marked.
  Now let $M$ be the set of nodes $m \in F_C$ such that every edge on the path
  from $m$ to the root of its component in $F_C$ is in $E$ and $m$ or its
  sibling in $F_C$ is marked.
  We say an edge is \emph{marked} if it belongs to the path from a node
  $m \in M$ to the root of its component, that is, if it is removed by
  fixing this node $m$.
  Next we prove that all edges in $E$ are marked.
  Since fixing a node or its sibling in $F_C$ results in the same forest
  and every node in $m$ is itself marked or has a marked sibling, this implies
  that there exists a subset of \emph{marked}
  potential exit nodes such that fixing them produces $F'$, that is,
  the refinement step applied to $F_C$ finds $F'$.

  Assume for the sake of contradiction that there is an unmarked edge in
  $E$.
  Since all ancestor edges of a marked edge are themselves marked,
  this implies that there exists a potential exit node $u \in \ecyc{F_C}$ whose
  parent edge $\edge{u}$ belongs to $E$ but is not marked, which in turn implies
  that neither $u$ nor its sibling $u'$ in $F_C$ is marked.
  The sequence of invocations that produce $F_C$ from $T_1$ and $T_2$ gives rise
  to a sequence of edges the algorithm cuts to produce $F_C$.
  For a step that cuts more than one edge, we cut these edges one by one.
  For Step~\ref{case:hyb:singleton} and branch~(i) of Case~\ref{case:hyb:cab},
  this ordering is chosen arbitrarily.
  For every branch of Step~\ref{case:hyb:non-sibling} that cuts an edge
  $\edge{x}$ with $x \in \set{a,c}$, we choose the ordering so that the parent
  edge of $x$ in $F_1$ is cut immediately after cutting $\edge{x}$ in $F_2$.
  Finally, in branch~(iii) of Case~\ref{case:hyb:cob}, we cut $\edge{c}$
  after $\edge{a}$.
  In the remainder of this proof, we use $F_1^i$ and $F_2^i$ to refer to the
  forests obtained from $T_1$ and $T_2$ after cutting the first $i$ edges.
  (This is a slight change of notation from the definition of~$F_C$,
  where we used $F_1^i$ and $F_2^i$ to denote the forests passed as arguments
  to the $i$th invocation.)
  Since $F_C$ is a refinement of $F_1^i$ and $F_2^i$, every node $x \in F_C$
  maps to the lowest node $y$ in $F_j^i$ such that the labelled descendant
  leaves of $x$ in $F_C$ are descendants of $y$ in $F_j^i$.
  This is analogous to the mappings $\phi_1(\cdot)$ and $\phi_2(\cdot)$ from
  $F_C$ to $T_1$ and $T_2$.
  To avoid excessive notation, we refer to the nodes in $F_1^i$ and $F_2^i$
  a node $x \in F_C$ maps to simply as~$x$.

  With this notation, the common parent $\parent{u}$ of $u$ and $u'$ in $F_C$
  is the lowest common ancestor of both nodes in any forest $F_j^i$.
  Since $u$ is a potential exit node of $F_C$, there is at least one
  hybrid edge in $\ecyc{F_C}$ induced by cutting a pendant edge of the path
  from $u$ to $p_u$ in some forest $F_j^i$.
  There may also be a hybrid edge induced by cutting a pendant edge of the
  path from $u'$ to $p_u$ in some forest~$F_j^i$.
  Either of these two types of edges are pendant to the path from $u$ to
  $u'$ in~$F_j^i$.
  Let $i$ be the highest index such that the $i$th edge we cut is pendant to
  the path from $u$ to $u'$ in $F_1^{i-1}$ or~$F_2^{i-1}$,
  and let $\edge{y}$ be this edge.
  Let $j \in \set{1, 2}$ so that we cut $\edge{y}$ in $F_j^{i-1}$.
  Since $u$ and $u'$ are siblings in $F_C$,
  the choice of index $i$ implies that $u$ and $u'$ are siblings in $F_1^i$
  and $F_2^i$.
  In particular, $y$ is the only pendant edge of the path from $u$ to $u'$ in
  $F_j^{i-1}$ and either $u$ or $u'$ is $y$'s sibling in~$F_j^{i-1}$.
  We use $x$ to refer to this sibling, and $x'$ to refer to $x$'s sibling in
  $F_C$ (that is, $x' = u'$ if $x = u$ and vice versa).
  We make two observations about $x$, $x'$, and $y$:
  \begin{enumerate}[label=(\roman{*}),leftmargin=0pt,itemindent=42pt]
  \item Since fixing a node in $F_C$ or its sibling produces the same forest,
    $F'$ can be obtained from $F_C$ by fixing a subset of nodes that
    includes $x$ or $x'$.
    In particular, $\aecutL{T_1, T_2, F_j^i \div \set{\edge{x}}}
    = \aecutL{T_1, T_2, F_j^i \div \set{\edge{x'}}}
    = \aecutL{T_1, T_2, F_j^i} - 1$, for $j \in \set{1, 2}$.
  \item Since edge $\edge{u}$ is not marked, neither $u$ nor $u'$ is marked,
    that is, $x$ is not marked in $F_C$ and, hence, $y$ is not tagged with
    ``$T_j$'' in $F_C$.
  \end{enumerate}
  Now we examine each of the steps that can cut $\edge{y}$ and prove that
  these observations lead to a contradiction.
  Thus, $E$ cannot contain an unmarked edge, and the lemma follows.

  First assume $\edge{y}$ belongs to $F_1^{i-1}$.
  Then $\edge{y}$ is cut by an application of Step~\ref{case:hyb:singleton} or
  $\edge{y}$ is the parent edge in $F_1^{i-1}$ of a node $y \in \set{a, c}$
  whose parent edge in $F_2^{i-2}$ is the $(i-1)$st edge we cut.
  First assume the former.
  Then $y$ is a root in $F_2^{i-1}$, which implies that there exists an $i' < i$
  such that the $i'$th edge we cut is an edge $\edge{z}$ in $F_2^{i'-1}$ such
  that $z$ is an ancestor of $y$ in $F_2^{i'-1}$ and $z$ is a node $b$ or $b_i$
  in this application of Step~\ref{case:hyb:non-sibling}.
  We choose the maximal such $i'$.
  This implies that no edge on the path from $y$ to $z$ is cut by any subsequent
  step.
  Indeed, if we cut such an edge in a forest $F_2^{i''-1}$, for $i' < i'' < i$,
  it would have to be an edge $\edge{z'}$ with $z' \in \set{a, c}$, by the
  choice of $i'$.
  If $z' = y$, then $\edge{y}$ would be cut in Step~\ref{case:hyb:non-sibling};
  if $z' \ne y$, then $\edge{y}$ would belong to a subtree of $F_1^{i''-1}$
  whose root is the member of a sibling pair, and $\edge{y}$ would never be cut.
  In either case, we obtain a contradiction.
  Now observe that any case of Step~\ref{case:hyb:non-sibling} that cuts an
  edge $\edge{b}$ or $\edge{b_i}$ tags $b$ or $b_i$ with
  ``$T_1$'', that is, $z$~is tagged with ``$T_1$'' immediately after cutting
  $\edge{z}$.
  Since we have just argued that no edges are cut on the path from $z$ to $y$
  and $y$ is a root in $F_2^{i-1}$, our rules for maintaining tags when
  suppressing nodes imply that $y$ inherits $z$'s ``$T_1$'' tag, a
  contradiction.

  Now suppose $\edge{y}$ belongs to $F_1^{i-1}$ and is cut in
  Step~\ref{case:hyb:non-sibling}, that is, $y \in \set{a, c}$.
  Since Step~\ref{case:hyb:non-sibling} cuts an edge in $F_1$ immediately
  after cutting the corresponding edge in $F_2$, the $(i-1)$st edge
  we cut is $y$'s parent edge in $F_2^{i-2}$.
  If $\edge{y}$ is cut by an application of Case~\ref{case:hyb:sc},
  assume w.l.o.g.\ that $y = c$ and, hence, $x = a$.
  Since the invocation $\aalgL{F_1^{i-2}, F_2^{i-2}, k}$ that cuts $\edge{y}$
  is viable and
  $\aecutL{T_1, T_2, F_2^{i-2} \div \set{\edge{x}}} =
  \aecutL{T_1, T_2, F_2^{i-2}} - 1$, the invocation
  $\aalgL{F_1^{i-2} \div \set{\edge{x}}, F_2^{i-2} \div \set{\edge{x}}, k - 1}$
  is also viable.
  Since we apply Case~\ref{case:hyb:sc}, $x$ and $x'$ are siblings in
  $F_C$, and $F_C$ is a refinement of $F_2^{i-2}$, we have
  $x' \reach[F_2^{i-2}] x \noreach[F_2^{i-2}] y$.
  Since $F_x$ is also a refinement of $F_2^{i-2}$,
  this implies that $x' \noreach[F_x] y$.
  In particular, $x'$ and $y$ are not siblings in $F_x$.
  Since $\edge{y}$ is the only pendant edge of the path from $x$ to $x'$
  in~$F_1^{i-2}$, this implies that either $y$ is a root
  in $F_x$ or its parent in $F_x$ is a proper ancestor in $F_1^{i-2}$ of the
  common parent of $x$ and $x'$ in $F_y = F_C$.
  In both cases, $d_1(y) < d_1(x)$, contradicting that we chose the invocation
  $\aalgL{F_1^{i-2} \div \set{\edge{y}}, F_2^{i-2} \div \set{\edge{y}}, k-1}$
  instead of the invocation $\aalgL{F_1^{i-2} \div \set{\edge{x}},
    F_2^{i-2} \div \set{\edge{x}}, k-1}$ on the path to~$F_C$.

  If $\edge{y}$ is cut by an application of Case~\ref{case:hyb:cob}, $y$ is
  tagged with ``$T_1$'' unless $y = a$ and we apply the third branch of this
  case, or $y = c$ and we apply the second branch of this case.
  If $y = a$ and we apply the third branch, then $x = c$ and
  the $(i+2)$nd edge we cut is edge $\edge{c}$ in $F_1^{i+1}$, which contradicts
  that $x = c$ has a sibling in $F_C$.
  If $y = c$ and we apply the second branch of this case, then $x = a$.
  However,
  since $\aecut{T_1, T_2, F_C \div \set{\edge{x}}} = \aecut{T_1, T_2, F_C} - 1$
  and we cut edge $\edge{y}$ to obtain $F_C$ from $F_2^{i-2}$, we have in fact
  $\aecutL{T_1, T_2, F_2^{i-2} \div \set{\edge{a}, \edge{c}}} =
  \aecutL{T_1, T_2, F_2^{i-2}} - 2$, that is, the invocation
  $\aalgL{F_1^{i-2} \div \set{\edge{a}, \edge{c}}, F_2^{i-2} \div
    \set{\edge{a}, \edge{c}}, k-2}$ is viable.
  This contradicts that we chose the invocation
  $\aalgL{F_1^{i-2} \div \set{\edge{c}}, F_2^{i-2} \div \set{\edge{c}}, k-1}$ as
  the next invocation on the path to~$F_C$.

  Finally, suppose $\edge{y}$ is cut by an application of
  Case~\ref{case:hyb:cab}.
  If $x'$ and $y$ are not siblings in $F_x$, then the same argument as for
  Case~\ref{case:hyb:sc} leads to a contradiction to the choice of $F_C$.
  So assume that $x'$ and $y$ are siblings in $F_x$, that is, that
  $d_1(x) = d_1(y)$.
  Since $\edge{y}$ is the last pendant edge of the path from $x$ to $x'$
  in either of the two forests $F_1$ and $F_2$, $x$ and $x'$
  are siblings in $F_2^{i-1}$.
  This implies that either $x$ and $x'$ are siblings also in $F_2^{i-2}$
  or $\edge{y}$ is the only pendant edge of the path from $x$ to $x'$
  in $F_2^{i-2}$.
  In the first case, we have $d_2(y) < d_2(x)$, contradicting that we chose
  the invocation
  $\aalgL{F_1^{i-2} \div \set{\edge{y}}, F_2^{i-2} \div \set{\edge{y}}, k-1}$ on
  the path to $F_C$, even though the invocation
  $\aalgL{F_1^{i-2} \div \set{\edge{x}}, F_2^{i-2} \div {\edge{x}}, k - 1}$ is
  viable.
  In the second case, cutting $\edge{y}$ in $F_2^{i-2}$ tags $y$ with
  ``$T_2$''.
  Since $y$ is the sibling of $x$ or $x'$ in $F_2^{i-2}$, this implies that
  $x$ or $x'$ is marked in~$F_C$, again a contradiction.

  Finally, assume $\edge{y}$ belongs to $F_2^{i-1}$.
  Then $\edge{y}$ is cut by an application of Case~\ref{case:hyb:cob} or
  Case~\ref{case:hyb:cab} because Case~\ref{case:hyb:sc} tags the bottom
  endpoint of each edge it cuts with ``$T_2$'', contradicting that $y$ is
  not tagged with~``$T_2$''.

  In Case~\ref{case:hyb:cob}, $\edge{y}$ is either $\edge{b}$ or $\edge{c}$
  because, when edge $\edge{a}$ is cut, $a$ is tagged with ``$T_2$''.
  First suppose $\edge{y} = \edge{b}$.
  Since $\edge{y}$ is the last pendant edge of the path from $x$ to $x'$ we cut
  in either of the two forests $F_1$ and $F_2$, we have $x=a$ and $x'=c$.
  However, since the current invocation $\aalgL{F_1^{i-1}, F_2^{i-1}, k}$
  is viable and $\aecutL{T_1, T_2, F_2^{i-1} \div \set{\edge{c}}} =
  \aecutL{T_1, T_2, F_2^{i-1} \div \set{\edge{x'}}} =
  \aecutL{T_1, T_2, F_2^{i-1}} - 1$, the invocation
  $\aalgL{F_1^{i-1} \div \set{\edge{c}}, F_2^{i-1} \div \set{\edge{c}}, k-1}$ is
  also viable, which contradicts that we chose the invocation
  $\aalgL{F_1^{i-1}, F_2^{i-1} \div \set{\edge{b}}, k-1}$ as the next invocation
  on the path to~$F_C$.

  If $\edge{y} = \edge{c}$, it must be an application of branch~(iii) of
  Case~\ref{case:hyb:cob} that cuts $\edge{y}$ because branch~(ii) tags $c$
  with ``$T_2$''.
  In this case, $x = b$ because we cut $\edge{a}$ before $\edge{c}$.
  Then, however, $b$ is $a$'s sibling in $F_2^{i-3}$ and the tail of $a$'s
  $T_2$-hybrid edge.
  Since $a$ is tagged with ``$T_2$'' in this case, this implies that $x = b$ is
  marked in $F_C$, a contradiction.

  In Case~\ref{case:hyb:cab} we tag $a$ or $c$ with ``$T_2$''.
  So $\edge{y}$ must be $\edge{b_h}$, for some pendant edge $\edge{b_h}$ of the
  path from $a$ to $c$ in $F_2^{i-q}$.
  Along with the fact that $\edge{y}$ is the last pendant edge of the path from
  $x$ to $x'$ we cut, this implies that $x=c$ or $x'=c$.
  Since the invocation $\aalgL{F_1^{i-q}, F_2^{i-q}, k}$ that cuts edges
  $b_1, b_2, \dots, b_q$ is viable and
  $\aecutL{T_1, T_2, F_2^{i-q} \div \set{\edge{x}}} =
  \aecutL{T_1, T_2, F_2^{i-q} \div \set{\edge{x'}}} =
  \aecutL{T_1, T_2, F_2^{i-q}} - 1$,
  the invocation $\aalgL{F_1^{i-q} \div \set{\edge{c}}, F_2^{i-q} \div
    \set{\edge{c}}, k-1}$ is also viable, contradicting that we chose the
  invocation $\aalgL{F_1^{i-q}, F_2^{i-q} \div
    \set{\edge{b_1}, \edge{b_2}, \dots, \edge{b_q}}, k-q}$
  as the next invocation on the path to~$F_C$.\qquad
\end{proof}

By Lemma~\ref{lem:f_c_refinement}, the algorithm returns ``yes'' if
$\dhyb{T_1, T_2} \le k_0$, and it cannot return ``yes'' if
$\dhyb{T_1, T_2} > k_0$.
Thus, our MAAF algorithm is correct.
Case~\ref{case:hyb:cob} makes an additional recursive call compared to the
algorithm from \S\ref{sec:naive_refinement}, but the number of recursive
calls in this case is still given by the recurrence $I(k) = 2I(k-1) + I(k-2)$,
which is also the worst case of Case~\ref{case:hyb:cab} in the MAF algorithm
(see Lemma~\ref{lem:invocations}).
Thus, the number of recursive calls made during the branching phase of the
algorithm remains $\OhL{2.42^{k_0}}$.
Since at most $k_0$ of the potential exit nodes of an AF $F$ found during the
branching phase are marked (one per root of $F$ other than $\rho$),
$\balg{F, k_0}$ takes $\OhL{2^{k_0} n}$ time to test whether fixing any subset
of these marked potential exit nodes yields an AAF of $T_1$ and $T_2$ with at
most $k_0 + 1$ components.
Thus, the total running time of the algorithm is
$\OhL{2.42^k\parensL{n + 2^k n}} = \OhL{4.84^k n}$, and we obtain the following
theorem.

\begin{theorem}
  \label{thm:maaf:fpt}
  For two rooted $X$-trees $T_1$ and $T_2$ and a parameter $k_0$, it
  takes $\OhL{4.84^{k_0} n}$ time to decide whether $\aecut{T_1, T_2, T_2} \le
  k_0$.
\end{theorem}

\subsection{Improved Refinement and Analysis}

\label{sec:combinations}

The algorithm we have developed so far finds a set of agreement
forests with marked potential exit nodes such that at least one of
these AFs $F$ can be refined to an MAAF $F'$ by fixing a subset of the
marked exit nodes in $F$. The algorithm then fixes every subset of these
marked potential exit nodes for each agreement forest it finds.
If $k'$ is the number of edges we cut to obtain $F$, there are $k'$ marked
potential exit nodes and $2^{k'}$ subsets of marked potential exit nodes to
check.
When $k'$ is small, the resulting time bound of $\OhL{2^{k'}n}$ for the
refinement step is substantially better than the bound of $\OhL{2^k n}$ obtained
using the naive upper bound of $k' \le k$ we used so far.
For large values of~$k'$, we observe that $F$ has $k'+1$
components because we always cut edges in a fully contracted forest (i.e., a
forest without degree-2 vertices other than its component roots).
When fixing a set of $k''$ potential exit nodes in the refinement step,
we cut at least $k''$ edges, and this increases the number of connected
components by at least $k''$, again because we cut edges along paths in fully
contracted forests.
Thus, if $k' + k'' > k$, we cannot possibly obtain an AAF with at most
$k+1$ components: the refinement step applied to $F$ needs to consider
only subsets of at most $k'' := k - k'$ potential exit nodes.
Since there are $k'$ marked potential exit nodes to choose from, this reduces
the running time of the refinement step applied to such a forest $F$ to
$\OhXL{\sum_{j=0}^{k''} \parensXL{\!\genfrac{}{}{0pt}{}{k'}{j}\!}n}$.
For large values of $k'$, $k''$~is small and the sum is significantly
less than $\OhL{2^{k'} n} = \OhL{2^k n}$.
Thus, we obtain a substantial improvement of the running time of the refinement
step also in this case, without affecting its correctness.
In summary, the only change to the MAAF algorithm from \S\ref{sec:marking}
we make in this section is to inspect all subsets of at most $k''$ marked
potential exit nodes in the refinement step, where $k'' := \min(k', k - k')$.

To analyze the running time of our algorithm using this improved refinement
step, we split each refinement step into several refinement steps.
A refinement step that tries all subsets of between $0$ and $k''$
marked potential exit nodes is replaced with $k''+1$ refinement steps:
for $0 \le j \le k''$, the $j$th such refinement step tries all subsets of
exactly $j$ marked potential exit nodes.
Its running time is therefore
$\OhXL{\parensXL{\!\genfrac{}{}{0pt}{}{k'}{j}\!} n}$, and the total cost
of all refinement steps remains unchanged.
Now we partition the refinement steps invoked for the different AFs
found during the branching phase into $k+1$ groups.
For $0 \le h \le k$, the $h$th group contains a refinement step applied to
an agreement forest $F$ if the number $k'$ of edges cut to obtain $F$ and
the size $j$ of the subsets of marked potential exit nodes the refinement
step tries satisfy $k' + j = h$.
We prove that the total running time of all refinement steps in the $h$th group
is $\OhL{3.18^h n}$.
Hence, the total running time of all refinement steps is
$\OhXL{\sum_{h=0}^k 3.18^h n} = \OhL{3.18^k n}$, which dominates the
$\OhL{2.42^k n}$ time bound of the branching phase, that is, the running time
of the entire MAAF algorithm is $\OhL{3.18^k n}$.

Now consider the tree of recursive calls made in the branching phase.
Since a given invocation $\aalg{F_1, F_2, k''}$ spawns further recursive calls
only if $F_2$ is not an AF of $T_1$ and $T_2$, and we invoke
the refinement step on $F_2$ only if $F_2$ is an AF of $T_1$ and
$T_2$, refinement steps are invoked only from the leaves of this recursion tree.
Moreover, since every refinement step in the $h$th group satisfies $k' + j = h$
and, hence, $k' \le h$, refinement steps in the $h$th group can be invoked
only for agreement forests that can be produced by cutting at most $h$
edges in $T_2$.
Thus, to bound the running time of the refinement steps in the $h$th group,
we can restrict our attention to the subtree of the recursion tree
containing all recursive calls $\aalg{F_1, F_2, k''}$ such that $F_2$
can be obtained from $T_2$ by cutting at most $h$ edges, that is,
$k'' \ge d := k - h$.
Since we want to obtain an upper bound on the cost of the refinement steps
in the $h$th group, we can assume that the shape of this subtree and the
set of refinement steps invoked from its leaves are such that the total cost
of the refinement steps is maximized.
We construct such a worst-case recursion tree for the refinement steps in
the $h$th group in two steps.

First we construct a recursion tree without refinement
steps and such that, for each $d \le k'' \le k$, the number of invocations
with parameter $k''$ in this tree is maximized.
As in the proof of Lemma~\ref{lem:invocations}, this is the case if each
recursive call with parameter $k'' \ge d + 2$ makes three recursive calls, two
with parameter $k''-1$ and one with parameter $k''-2$, and each recursive
call with parameter $k'' = d+1$ makes two recursive calls with parameter
$k''-1$.
As in the proof of Lemma~\ref{lem:invocations}, this implies that every
recursive call with parameter $k''$ has a tree of
$\ThehtaXL{\parensL{1 + \sqrt{2}}^{k''-d}}$ recursive calls below it,
and the size of the entire tree is
$\OhXL{\parensL{1 + \sqrt{2}}^{k-d}} = \OhXL{\parensL{1 + \sqrt{2}}^h}$.
The second step is to choose a subset of recursive calls in this tree
for which we invoke the refinement step \emph{instead of} spawning further
recursive calls, thereby turning them into leaves.
In effect, for each such node with parameter $k''$, we replace its subtree
of $\ThehtaXL{\parensL{1 + \sqrt{2}}^{k''-d}}$ recursive calls with a single
refinement step of cost $\OhXL{\parensXL{\!\genfrac{}{}{0pt}{}{k'}{j}\!}n}$,
where $k' := k - k'' = h + d - k''$ and $j := h - k' = k'' - d$.
By charging the cost of this refinement step equally to the nodes in the removed
subtree, each node in this subtree is charged a cost of
$\ThehtaXL{\parensXL{\!\genfrac{}{}{0pt}{}{k'}{j}\!}n/\parensL{1 + \sqrt{2}}^{k''-d}} =
\ThehtaXL{\parensXL{\!\genfrac{}{}{0pt}{}{k'}{j}\!}n/\parensL{1 + \sqrt{2}}^j}$.
The total running time of all refinement steps in the $h$th group is the sum of
the charges of all nodes removed from the recursion tree.
Since we can remove at most $\OhXL{\parensL{1+\sqrt{2}}^h}$ nodes from the tree,
the cost of all refinement steps in the $h$th group is therefore
\begin{equation}
  \label{eq:rt1}
  \OhV{\parensL{1+\sqrt{2}}^h\frac{\parensXL{\!\genfrac{}{}{0pt}{}{k'}{j}\!}n}{\parensL{1+\sqrt{2}}^{j}}}
  = \OhV{\parensL{1+\sqrt{2}}^{k'} \binom{k'}{j}n},
\end{equation}
where $k'$ and $j$ are chosen so that
$\parensXL{\!\genfrac{}{}{0pt}{}{k'}{j}\!}/\parensL{1+\sqrt{2}}^j$ is
maximized subject to the constraints $0 \le j \le k'$ and
$k' + j = h$.
It remains to bound this expression by $\OhL{3.18^h n}$.
First assume that $k' \le 2h/3$.
Then we can bound $\parensXL{\!\genfrac{}{}{0pt}{}{k'}{j}\!}$ by $2^{k'}$, and
$\parensL{1+\sqrt{2}}^{k'} \cdot \parensXL{\!\genfrac{}{}{0pt}{}{k'}{j}\!}$ by
$\parensL{2 + 2\sqrt{2}}^{k'} \le 4.84^{2h/3} \le 2.87^h$, that is,
(\ref{eq:rt1})~is bounded by $\OhL{2.87^h n}$.
For $k' = h$, we have $j = 0$ and, hence, (\ref{eq:rt1}) is bounded
by $\OhL{2.42^h n}$ in this case.
To bound (\ref{eq:rt1}) for $2h/3 < k' < h$, we make use of the following
observation.

\begin{observation}
  \label{obs:choose}
  $\displaystyle\binom{x}{y} =
  \OhV{\parensV{\frac{x}{y}}^y \parensV{\frac{x}{x-y}}^{x-y}}.$
\end{observation}

Observation~\ref{obs:choose} allows us to bound (\ref{eq:rt1}) by
\begin{multline*}
  \OhV{\parensL{1+\sqrt{2}}^{k'}
  \parensV{\frac{k'}{j}}^j \parensV{\frac{k'}{k' - j}}^{k' - j} n}
  =\\
  \OhV{\parensV{\parensL{1+\sqrt{2}}^\alpha
  \parensV{\frac{\alpha}{1-\alpha}}^{1 - \alpha}
  \parensV{\frac{\alpha}{2\alpha-1}}^{2\alpha-1}}^h n},
\end{multline*}
where $\alpha := k'/h$ and, hence, $k' = \alpha h$ and $j = (1 - \alpha)h$.
It remains to determine the value of $\alpha$ such that $2/3 < \alpha < 1$
and the function
\begin{equation*}
  b(\alpha) = \parensL{1 + \sqrt{2}}^\alpha \parensV{\frac{\alpha}{1-\alpha}}^{1-\alpha}
  \parensV{\frac{\alpha}{2\alpha-1}}^{2\alpha-1}
\end{equation*}
is maximized.
Taking the derivative and setting to zero, we obtain that $b(\alpha)$ is
maximized for
$\alpha = \frac{1}{2} + \frac{\sqrt{7 + 6\sqrt{2}}}{10 + 2\sqrt{2}}$, which
gives $b(\alpha) \le 3.18$.
This finishes the proof that the total cost of the refinement steps in the
$h$th group is $\OhL{3.18^h n}$, which, as we argued already, implies that
the running time of the entire algorithm is $\OhL{3.18^k n}$.
Thus, we have the following theorem.

\begin{theorem}
  \label{thm:final-hyb}
  For two rooted $X$-trees $T_1$ and $T_2$ and a parameter $k$, it takes
  $\OhL{3.18^k n}$ time to decide whether $\aecut{T_1,T_1,T_2} \le k$.
\end{theorem}

As with the MAF algorithms, we can use known
kernelization rules~\cite{bordewich07chn}
to transform the trees $T_1$ and $T_2$ into two trees $T_1'$ and $T_2'$
of size $\Oh{\ecut{T_1, T_2, T_2}}$.
However, unlike the kernelization rules used for SPR distance, these
kernelization rules produce trees that do \emph{not} have the same hybridization
number as $T_1$ and~$T_2$.
One of these rules, the \emph{Chain Reduction}, replaces a chain of leaves
$a_1, a_2, \ldots$ with a pair of leaves $a,b$.
Bordewich and Semple~\cite{bordewich07chn} showed that in an MAAF of the
resulting two trees, either $a$ and $b$ are both isolated or neither is.
A corresponding MAAF of $T_1$ and $T_2$ can be obtained by cutting the parent
edges of $a_1, a_2, \ldots$ in the first case or replacing $a$ and $b$ with
the sequence of leaves $a_1, a_2, \ldots$ in the second case.
The difference in size between these two MAAFs is captured by assigning
the number of leaves removed by the reduction as a weight to the pair $(a,b)$.
The weight of an AAF of the two reduced trees $T_1'$ and $T_2'$ then is
the number of components of the AAF plus the weights of all such pairs $(a,b)$
such that $a$ and $b$ are isolated in the AAF.
This weight equals the size of the corresponding AAF of $T_1$ and $T_2$.

It is not difficult to incorporate these weights into our MAAF algorithm.
Whenever the refinement algorithm would return ``yes'',
we first add the sum of the weights of isolated pairs to the
number of components in the found AAF.
If, and only if, this total is less than or equal to~$k_0$, we return ``yes''.
Any AF $F$ of $T_1'$ and $T_2'$ with weight $w(F) = \aecut{T_1,T_2,T_2}$
has at most $w(F)$ components and thus will be examined by this strategy.
Similarly, the depth of the recursion is bounded by the number of components,
and thus by $k_0$.
Thus, we obtain the following corollary.

\begin{corollary}
  \label{cor:maaf}
  For two rooted $X$-trees $T_1$ and $T_2$ and a parameter $k$, it
  takes $\OhL{3.18^k k + n^3}$ time to decide whether $\aecut{T_1, T_2,
    T_2} \le k$.
\end{corollary}

\section{Conclusions}

\label{sec:concl}

The algorithms presented in this paper are the theoretically fastest algorithms
for computing SPR distances and hybridization numbers of rooted phylogenies.
The most important open problem is extending our approach to computing maximum
agreement forests and maximum acyclic agreement forests for multifurcating trees
and for more than two trees.
Evolutionary biologists often construct phylogenetic trees using methods that
assign a measure of statistical support to each edge of the tree.
Contracting edges with poor statistical support eliminates bipartitions that
may be artifacts of the manner in which the tree was constructed
but the resulting trees are multifurcating trees.
If we can extend our methods to support multifurcating trees, the comparisons
of binary phylogenies our new algorithms make possible can be applied also
to multifurcating trees.
The kernelization results of Linz and Semple~\cite{linz09hnt}
for maximum acyclic agreement forests apply to such trees.
Extending our bounded search tree approach to computing agreement forests
of multifurcating trees is currently the focus of ongoing efforts on our part.

A first step towards comparing \emph{multiple} phylogenies over a set of species
could be to identify groups of species whose pattern of relatedness is the same
in all trees, which is exactly what a maximum agreement forest of all
the trees in the given set would represent.
The $8$-approximation algorithm by Chataigner~\cite{chataigner05} for computing
an MAF of two or more unrooted phylogenies and the FPT algorithm by Chen
and Wang~\cite{chen12} for computing all MAAFs of a set of rooted phylogenies
are important steps in this direction.
We believe that some of the ideas in this paper may lead to improvements
of the latter result.

While the theoretical results presented in this paper are interesting in their
own right, as they shed further light on the complexity of computing agreement
forests, experimental results indicate that our algorithms also
perform very well in practice.
In~\cite{whidden2010fast}, we evaluated the practical performance
of our algorithms for rooted SPR distance and demonstrated that they
are an order of magnitude faster than the currently best exact alternatives
\cite{wu2009practical,bonet2009efficiently} based on reductions to
integer linear programming and satisfiability testing, respectively.
The implementation and its source code are publicly available~\cite{rspr}.
The largest distances reported using implementations of previous
methods are a hybridization number of 19 on 46 taxa~\cite{wu2010fast}
and an SPR distance of 19 on 46 taxa~\cite{wu2009practical}.
In contrast, our method took less than 5
hours to compute SPR distances of up to 46 on trees with 144 taxa and
99 on synthetic 1000-leaf trees and required less than one second
on average to compute SPR distances of up to 19 on 144 taxa.
This represents a major step forward towards tools that can infer
reticulation scenarios for the thousands of genomes that have been
fully sequenced to date.


\bibliographystyle{siam}
\bibliography{rspr}

\end{document}